%% file: main.tex
\documentclass[10pt]{article}
\usepackage{times}

\usepackage{multicol}
\usepackage{latexsym}
\usepackage{amstext,amssymb,amsmath,amsthm}
\usepackage{amsbsy}
\usepackage{fullpage}
\usepackage{paralist}

\usepackage[pdflatex=true,recompilepics=auto]{gastex} 
\usepackage{pstricks}
\usepackage{pict2e}
\usepackage{graphicx, subfigure}

\usepackage[all]{xy}

\newtheorem{thm}{Theorem} 

\newtheorem{lem}{Lemma}
\newtheorem{prop}{Proposition}
\newtheorem{rem}{Remark}
\newtheorem{cor}{Corollary}
\newtheorem{Obs}{Observation}

\newcommand{\Z}{\mathbb{Z}}
\newcommand{\Q}{\mathbb{Q}}
\newcommand{\R}{\mathbb{R}}

\newcommand{\Set}[1]{\{ #1 \}}
\newcommand{\Range}[2]{#1 ,\dots, #2}
\newcommand{\RangeSet}[2]{\Set{ \Range{#1}{#2}}}

\newcommand{\Tuple}[1]{\langle #1 \rangle}

\newcommand{\AttrTwo}[1]{\ensuremath{Attr_2(#1)}}

\newcommand{\Heading}[1]{\vspace{-0.25cm}\paragraph{\bf{#1}}}
\newcommand{\BeginProof}{\vspace{-0.25cm}\begin{proof}}

\newcommand{\Comment}[1]{}

\newcommand{\Avg}{\mathsf{Avg}}
\newcommand{\LimInfAvg}{\mathsf{LimInfAvg}}
\newcommand{\LimSupAvg}{\mathsf{LimSupAvg}}
\newcommand{\LimAvg}{\mathsf{LimAvg}}

\newcommand{\VEC}{\vec} 

\def\Win{\mathsf{Win}}

\newcommand{\Nat}{\ensuremath{\mathbb{N}}}

\def\Skip{\mathit{skip}}
\def\Pop{\mathit{pop}}
\def\Push{\mathit{push}}
\def\Com{\mathsf{Com}}
\def\com{\mathit{com}}
\def\SH{\mathsf{SH}}
\def\ASH{\mathsf{ASH}}
\def\Top{\mathsf{Top}}
\newcommand{\wps}{\mathcal{A}}
\newcommand{\atuple}[1]{\langle #1 \rangle}
\newcommand{\ov}{\overline}
\newcommand{\Gr}{\mathsf{Gr}}

\def\Calls{\mathsf{Calls}}
\def\Returns{\mathsf{Retns}}
\newcommand{\LocalHistory}{\mathsf{LocalHistory}}

\newcommand{\wrg}{\mathcal{A}}
\newcommand{\En}{\mathit{En}}
\newcommand{\Ex}{\mathit{Ex}}

\newcommand{\Sig}{\mathsf{Sig}}

\newcommand{\NumModules}{\mathsf{M}}
\newcommand{\NumExits}{\mathsf{E}}
\newcommand{\NumBoxes}{\mathsf{B}}

\makeatletter

\begingroup \catcode `|=0 \catcode `[= 1
\catcode`]=2 \catcode `\{=12 \catcode `\}=12
\catcode`\\=12 |gdef|@xcomment#1\end{comment}[|end[comment]]
|endgroup

\def\@comment{\let\do\@makeother \dospecials\catcode`\^^M=10\def\par{}}

\def\begincomment{\@comment\@xcomment}

\makeatother

\newenvironment{comment}{\begincomment}{}

\begin{document}

\title{Hyperplane Separation Technique for\\Multidimensional Mean-Payoff Games}
\author{Krishnendu Chatterjee (IST Austria) \and 
Yaron Velner (Tel Aviv University, Israel)}
\date{}
\maketitle
\pagestyle{plain}

\input{abstract-intro}

\input{finite-game}

\input{wps-multi}

\input{modular}



\smallskip\noindent{\bf Concluding remarks.}
In this work we considered the fundamental algorithmic 
questions related to multidimensional mean-payoff objectives
in finite-state games, pushdown graphs, and pushdown games.
We presented algorithms that precisely characterized the parameters 
that need to be constant for polynomial-time algorithms.
Moreover, we also established the hardness of fixed parameter tractability 
for the relevant problems.

{\small
\smallskip\noindent{\bf Acknowledgement.}
The research was supported by Austrian Science Fund (FWF) Grant No P 23499-N23,
FWF NFN Grant No S11407-N23 (RiSE), ERC Start grant (279307: Graph Games), 
Microsoft faculty fellows award, the RICH Model Toolkit (ICT COST Action IC0901),
and was carried out in partial fulfillment of the requirements for the Ph.D. degree of 
the second author.
}

{\footnotesize
\bibliographystyle{plain}
\bibliography{diss}
}


\end{document}

%% file: abstract-intro.tex
\begin{abstract}
Two-player games on graphs are central in many problems in formal verification 
and program analysis such as synthesis and verification of open systems. 
In this work, we consider both finite-state game graphs, and recursive game 
graphs (or pushdown game graphs) that can model the control flow of sequential 
programs with recursion.
The objectives we study are multidimensional mean-payoff objectives, where 
the goal of player~1 is to ensure that the mean-payoff is at least zero in all
dimensions.
In pushdown games two types of strategies are relevant: (1)~global strategies, 
that depend on the entire global history; and (2)~modular strategies, that 
have only local memory and thus do not depend on the context of invocation.
We present solutions to several fundamental algorithmic questions and 
our main contributions are as follows:
(1)~We show that finite-state multidimensional mean-payoff games can be 
solved in polynomial time if the number of dimensions and the maximal absolute 
value of weights are fixed; whereas if the number of dimensions is arbitrary, 
then the problem is known to be coNP-complete.
(2)~We show that pushdown graphs (or one-player pushdown games) with
multidimensional mean-payoff objectives can be solved in polynomial
time.
For both (1) and (2) our algorithms are based on hyperplane separation technique.
(3)~For pushdown games under global strategies both one and multidimensional
mean-payoff objectives problems are known to be undecidable, and we show
that under modular strategies the multidimensional problem is also undecidable;
under modular strategies the one-dimensional problem is known to be
NP-complete. 
We show that if the number of modules, the number of exits, and the 
maximal absolute value of the weights are fixed, then pushdown games
under modular strategies with one-dimensional mean-payoff objectives
can be solved in polynomial time, and if either the number of exits
or the number of modules is unbounded, then the problem is NP-hard.
(4)~Finally we show that a fixed parameter tractable algorithm for 
finite-state multidimensional mean-payoff games or pushdown games 
under modular strategies with one-dimensional mean-payoff objectives
would imply the solution of the long-standing open problem of 
fixed parameter tractability of parity games.
\end{abstract}

\newpage
\setcounter{page}{1}

\section{Introduction}
In this work we present a hyperplane based technique that solves several fundamental algorithmic open questions
for multidimensional mean-payoff objectives.
We first present an overview of mean-payoff games, then the 
important extensions, followed by the open problems, and finally 
our contributions.

\smallskip\noindent{\bf Mean-payoff games on graphs.} 
Two-player games played on finite-state graphs provide the mathematical 
framework to analyze several important problems in computer science as 
well as in mathematics, such as formal analysis of reactive systems~\cite{BuchiLandweber69,RamadgeWonham87,PnueliRosner89}.
Games played on graphs are dynamic games that proceed for an infinite 
number of rounds.
The vertex set of the graph is partitioned into player-1 vertices
and player-2 vertices. 
The game starts at an initial vertex, and if the current vertex is a 
player-1 vertex, then player~1 chooses an outgoing edge, and 
if the current vertex is a player-2 vertex, then player~2 does 
likewise.
This process is repeated forever, and gives rise to 
an outcome of the game, called a {\em play}, that consists of the 
infinite sequence of vertices that are visited. 
The most well-studied payoff criteria in such games  is the \emph{mean-payoff} 
objective, where a weight (representing a reward) is associated with every 
transition and the goal of one of the players is to maximize the long-run 
average of the weights;  and the goal of the opponent is to minimize.
Mean-payoff games and the special case of graphs (with only one player) with mean-payoff
objectives have been extensively studied over the last three decades; 
e.g. \cite{Karp78,EM79,ZP95,GKK88}.
Graphs with mean-payoff objectives can be solved in polynomial time~\cite{Karp78}, 
whereas mean-payoff games can be decided in NP $\cap$ coNP~\cite{EM79,ZP95}. 
The mean-payoff games problem is an intriguing problem and one of the rare combinatorial problems that is
known to be in NP $\cap$ coNP, but no polynomial time algorithm is known.
However, pseudo-polynomial time algorithms exist for mean-payoff games~\cite{ZP95,BCDGR11}, 
and if the weights are bounded by a constant, then the algorithm is polynomial.

\smallskip\noindent{\bf The extensions.}
Motivated by applications in formal analysis of reactive systems, 
the study of mean-payoff games has been extended in two directions:
(1)~pushdown mean-payoff games; and 
(2)~multidimensional mean-payoff games on finite game graphs.
Pushdown games, aka  games on recursive state machines, can model reactive 
systems with recursion (i.e., model the control flow of sequential programs 
with recursion).
Pushdown games have been studied widely with applications in verification, 
synthesis, and program analysis in~\cite{Wal01,Wal00,AlurReach,AlurParity} 
(also see~\cite{EY05,EY09,BBKO11,BBFK08} for sample research in stochastic pushdown games).
In applications of verification and synthesis, the quantitative objectives
that typically arise are multidimensional quantitative objectives 
(i.e., conjunction of several objectives), e.g., to express 
properties like the average response time between a grant and a request is 
below a given threshold $\nu_1$, and the average number of unnecessary grants 
is below a threshold $\nu_2$. 
Thus mean-payoff objectives can express properties related to 
resource requirements, performance, and robustness; 
multiple objectives can express the different, potentially dependent or 
conflicting objectives.
Moreover, recently many quantitative logics and automata theoretic formalisms
have been proposed with mean-payoff objectives in their heart to express 
properties such as reliability requirements, and resource bounds of reactive
systems~\cite{CDH10,BCHK11,DM10,BBFR13}, and several quantitative synthesis
questions (such as synthesis from incompatible specifications~\cite{CGHRT12}) 
translate directly to multidimensional mean-payoff games. 
Thus pushdown games and graphs with mean-payoff objectives, and finite-state 
game graphs with multidimensional mean-payoff objectives are fundamental  
theoretical questions in model checking of quantitative logics and 
quantitative analysis of reactive systems (along with recursion features).
Pushdown games with multidimensional objectives are also a natural 
generalization to study.
Furthermore, in applications related to reactive system analysis, the number of 
dimensions of mean-payoff objectives is typically small, say~2 or~3, 
as they denote the different types of resources; and the weights 
denoting the resource consumption amount are also bounded
by constants; whereas the state space of the reactive system is huge; 
see~\cite{BCHJ09,BCKN12} for examples.

\smallskip\noindent{\bf Relevant aspects of pushdown games.} 
In pushdown games two types of strategies are relevant and studied in 
the literature. 
The first one are the \emph{global} strategies, where a global strategy can 
choose the successor vertex depending on the entire global history of the 
play; where history is the finite sequence of configurations of the current 
prefix of a play.
The second are \emph{modular} strategies, which are 
understood more intuitively in the model of games on recursive state machines.
A \emph{recursive state machine} (RSM) consists of a set of component machines 
(or modules). 
Each module has a set of \emph{nodes} (atomic states) and \emph{boxes} 
(each of which is mapped to a module), a well-defined interface consisting of 
\emph{entry} and \emph{exit} nodes, and edges connecting nodes/boxes. 
An edge entering a box models the invocation of the module associated with the 
box and an edge leaving the box represents return from the module.
In the game version the nodes are partitioned into player-1 nodes and 
player-2 nodes.
Due to recursion the underlying global state-space is infinite and isomorphic
to pushdown games. 
The equivalence of pushdown games and recursive games has been established
in~\cite{AlurReach}. 
A modular strategy is a strategy that has only local memory, and thus, 
the strategy does not depend on the context of invocation of the module,
but only on the history within the current invocation of the module. 
Informally, modular strategies are appealing because they are 
stackless strategies, decomposable into one for each module.

\smallskip\noindent{\bf Previous results and open questions.}
We now summarize the main previous results and open questions and then 
present our contributions.
\begin{compactenum}
\item \emph{(Finite-state graphs).} 
Finite-state graphs (or one-player games) with mean-payoff objectives can be solved
in polynomial time~\cite{Karp78}, and finite-state graphs with multidimensional 
mean-payoff objectives can also be solved in polynomial time~\cite{VR11}
using the techniques to detect zero-circuits in graphs of~\cite{KS88}.

\item \emph{(Finite-state games).}
Finite-state games with a one-dimensional mean-payoff objective can be decided in 
NP $\cap$ coNP~\cite{ZP95,EM79}, and pseudo-polynomial time 
algorithms exist for mean-payoff games~\cite{ZP95,BCDGR11}: the current
fastest known algorithm works in time $O(n \cdot m \cdot W)$, 
where $n$ is the number of vertices, $m$ is the number of edges, and 
$W$ is the maximal absolute value of the weights~\cite{BCDGR11}.
Finite-state games with multidimensional mean-payoff objectives are 
coNP-complete with weights in $\Set{-1,0,1}$ (i.e., the weights are bounded
by a constant) 
but with arbitrary dimensions~\cite{CDHR10}, and the current best known 
algorithm works in time $O(2^n \cdot \text{poly}(n,m,\log W))$.

\item \emph{(Pushdown graphs and games).}
Pushdown graphs and games have been studied only for one-dimensional mean-payoff
objectives~\cite{CV12}.
Under global strategies, pushdown graphs with a one-dimensional mean-payoff objective
can be solved in polynomial time, whereas pushdown games are undecidable.
Under modular strategies, pushdown graphs with every module restricted to 
have single exit and weights restricted to $\Set{-1,0,1}$ are NP-hard, and 
pushdown games with any number of exits and general weight function 
are in NP (i.e., the problems are NP-complete)~\cite{CV12}.
\end{compactenum}
Many fundamental algorithmic questions have remained open for analysis of 
finite-state and pushdown graphs and games with multidimensional mean-payoff 
objectives where the goal of player~1 is to ensure that the mean-payoff is at 
least zero in all dimensions.
The most prominent ones are: 
(A)~Can finite-state game graphs with multidimensional mean-payoff objectives
with $2$ or $3$ dimensions and constant weights be solved in polynomial time? 
(note that with arbitrary dimensions the problem is coNP-complete, and for arbitrary
weights no polynomial time algorithm is known even for the one-dimensional case);
(B)~Can pushdown graphs under global strategies with multidimensional mean-payoff 
objectives be solved in polynomial time?; 
(C)~Can a polynomial time algorithm be obtained for pushdown games under modular 
strategies with a one-dimensional mean-payoff objective when relevant parameters 
(such as the number of modules) are bounded?; and 
(D)~In what complexity class does pushdown games under modular strategies with 
multidimensional mean-payoff objectives lie?
The above questions are not only of theoretical interest, but stem from 
practically motivated problems of formal analysis of reactive systems.

\smallskip\noindent{\bf Our contributions.} 
In this work we present a hyperplane separation technique to provide answers 
to many of the open fundamental questions. 
Our contributions are summarized as follows:

\begin{compactenum}

\item \emph{(Hyperplane technique).} 
We use the separating hyperplane technique from computational geometry 
to answer the open questions (A) and (B) above.
First, we present an algorithm for finite-state games with multidimensional 
mean-payoff objectives of $k$-dimensions that works in time 
$O(n^2 \cdot m \cdot k \cdot W \cdot (k \cdot n \cdot W)^{k^2 +2\cdot k +1})$ 
(Section~2: Theorem~\ref{thm1}), 
and thus for constant weights and any constant $k$ (not only $k=2$ or $k=3$) 
our algorithm is polynomial.
Second, we present a polynomial-time algorithm for pushdown graphs under global 
strategies with multidimensional mean-payoff objectives 
(Section~3: Theorem~\ref{thm:LimInfInP}); the algorithm is polynomial for 
general weight function and any number of dimensions.
Our key intuition is to reduce the multidimensional problem to 
searching for a separating hyperplane such that all realizable mean-payoff vectors lie on one side of the hyperplane. 
This intuition allows us to search for a vector, 
which is normal to the separating hyperplane,  
and reduce the multidimensional problem to a one-dimensional problem by
multiplying the multidimensional weight function by the vector.

\item \emph{(Modular pushdown games).} 
We first show that the hyperplane techniques do not extend for modular 
strategies in pushdown games: we show that pushdown games under modular 
strategies with multidimensional mean-payoff objectives with fixed number of 
dimensions are undecidable (Section~4: Theorem~\ref{thm:undec}).
Thus the only relevant algorithmic problem for pushdown games is the modular
strategies problem for a one-dimensional mean-payoff objective; 
under global strategies
even a one-dimensional mean-payoff objective problem is undecidable~\cite{CV12}.
It was already shown in~\cite{CV12} that if the number of modules is unbounded,
then even with single exits for every module the problem is NP-hard.
We show that pushdown games under modular strategies with one-dimensional mean-payoff 
objectives are NP-hard with two modules and with weights $\Set{-1,0,1}$ 
if the number of exits is unbounded (Section~4: Theorem~\ref{thm:np-hard}).
Thus to obtain a polynomial time algorithm we need to bound both the number of modules
as well as the number of exits. 
We show that pushdown games under modular strategies with one-dimensional mean-payoff 
objectives can be solved in time $(n \cdot \NumModules)^{O(\NumModules^5 + \NumModules \cdot \NumExits^2)} \cdot W^{O(\NumModules^2 + \NumExits)}$,
where $n$ is the number of vertices, 
$W$ is the maximal absolute weight, $\NumModules$ is the number of modules, and $\NumExits$ is the number
of exits (Section~4: Theorem~\ref{thm:SolveBoundedRecursiveGamesWith}).
Thus if $\NumModules$, $\NumExits$, and $W$ are constants, our algorithm is polynomial.
Hence we answer the open questions (C) and (D).

\item \emph{(Hardness for fixed parameter tractability).}
Given our polynomial-time algorithms when the parameters are fixed 
for finite-state multidimensional mean-payoff games and pushdown 
games with a one-dimensional mean-payoff objective under modular strategies, 
a natural question is whether they are fixed parameter tractable, e.g.,
could we obtain an algorithm that runs in time 
$f(k) \cdot O(\text{poly}(n,m,W))$ (resp. $f(\NumModules,\NumExits) \cdot O(\text{poly}(n,W))$)  
for finite-state multidimensional mean-payoff games (resp. for pushdown modular 
games with one-dimensional objective), for some computable function 
$f$ (e.g., exponential or double exponential).
We show the hardness of fixed parameter tractability problem by reducing the long-standing
open problem of fixed parameter tractability of parity games to both the 
problems (Section~2: Theorem~\ref{thm:finite-reduction} and 
Section~4: Theorem~\ref{thm:parity-modular-reduction}), 
i.e., fixed parameter tractability of any of the above problems
would imply fixed parameter tractability of parity games.

\end{compactenum}


\begin{comment}

\smallskip\noindent{\bf Technical contributions.} 
Our key technical contributions are as follows.
For pushdown systems under global strategies we show that the mean-payoff objective problem 
can be solved by only considering additional stack height that is polynomial.
We then show that the stack height bounded problem can be solved in 
polynomial time using a dynamic programming style algorithm.
For pushdown games under global strategies our undecidability result is obtained as a reduction 
from the universality problem of \emph{weighted} automata (which is 
undecidable~\cite{Krob,ABK11}).
For modular strategies we first show the existence of cycle independent 
modular winning strategies, and then show memoryless modular strategies 
are sufficient. 
Given memoryless modular strategies and polynomial time algorithm for
pushdown systems, we obtain the NP upper bound for the modular strategies
problem.
Our NP-hardness result for modular strategies is a reduction from the
3-SAT problem.
\end{comment}

%% file: finite-game.tex
\newcommand{\Out}{\mathsf{Out}}
\newcommand{\strab}{\sigma}

\newcommand{\straa}{\tau}
\newcommand{\set}[1]{\{#1\}}
\newcommand{\sseq}{\langle v_0,v_1,v_2,\ldots\rangle}

\newcommand{\pat}{\pi}
\newcommand{\Pat}{\Pi}
\newcommand{\wh}{\widehat}

\section{Finite-State Games with Multidimensional Mean-Payoff Objectives}\label{sec:finite-games}

In this section we will present two results: (1)~an algorithm for finite-state
multidimensional mean-payoff games for which the running time is polynomial 
when the number of dimensions and weights are fixed;
(2)~a reduction of finite-state parity games to finite-state multidimensional 
mean-payoff games with polynomial weights and arbitrary dimensions that shows that 
fixed parameter tractability of multidimensional mean-payoff games would imply 
the solution of a long-standing open problem of fixed parameter tractability of 
parity games. 
We start with the basic definitions of finite-state games, strategies, and 
mean-payoff objectives.

\smallskip\noindent{\bf Game graphs.} 
A \emph{game graph} $G=((V,E),(V_1,V_2))$ consists of a \emph{finite} 
directed graph $(V,E)$ with a finite set $V$ of $n$ vertices and a set $E$ of 
$m$ edges, 
and a partition $(V_1,V_2)$ of $V$ into two sets.
The vertices in $V_1$ are {\em player-1 vertices}, where player~1 chooses the
outgoing edges, and the vertices in $V_2$ are {\em player-2 vertices},
where  player~2 (the adversary to player~1) chooses the outgoing edges.
Intuitively game graphs are the same as AND-OR graphs.
For a vertex $u\in V$, we write $\Out(u)=\set{v\in V \mid (u,v) \in E}$ 
for the set of successor vertices of~$u$.
We assume that every vertex has at least one outgoing edge, i.e., 
$\Out(u)$ is non-empty for all vertices $u\in V$.

\smallskip\noindent{\bf Plays.} A game is played by two players: 
player~1 and player~2, who form an infinite path in the game graph by 
moving a token along edges.
They start by placing the token on an initial vertex, and then they
take moves indefinitely in the following way.
If the token is on a vertex in~$V_1$, then player~1 moves the token along
one of the edges going out of the vertex.
If the token is on a vertex in~$V_2$, then player~2 does likewise.
The result is an infinite path in the game graph, called {\em plays}.
Formally, a \emph{play} is an infinite sequence 
$\pat=\sseq$ of vertices such that $(v_j,v_{j+1}) \in E$ for all $j \geq 0$.

\smallskip\noindent{\bf Strategies.} 
A strategy for a player is a rule that specifies how to extend plays.
Formally, a \emph{strategy} $\straa$ for player~1 is a function 
$\straa$: $V^* \cdot V_1 \to V$ that, given a finite sequence of vertices 
(representing the history of the play so far) which ends in a player~1 
vertex, chooses the next vertex.
The strategy must choose only available successors, i.e., for all $w \in V^*$ 
and $v \in V_1$ we have $\straa(w \cdot v) \in \Out(v)$.
The strategies for player~2 are defined analogously.
A strategy is \emph{memoryless} if it is independent of the history and 
only depends on the current vertex. 
Formally, a memoryless strategy for player~1 is a function 
$\straa$: $V_1 \to V$ such that $\straa(v)\in \Out(v)$ for all $v \in V_1$, and 
analogously for player~2 strategies.
Given a starting vertex $v\in V$, a strategy $\straa$ for player~1, 
and a strategy $\strab$ for player~2, there is a unique play, 
denoted $\pat(v,\straa,\strab)=\sseq$, which is defined as follows: 
$v_0=v$ and for all $j \geq 0$,
if $v_j \in V_1$, then $\straa((v_0,v_1,\ldots v_j))=v_{j+1}$, and
if $v_j \in V_2$, then $\strab((v_0,v_1,\ldots,v_j))=v_{j+1}$.

\smallskip\noindent{\bf Graphs obtained under memoryless strategies.}
A player-1 graph is a special case of a game graph where all vertices
in $V_2$ have a unique successor (and player-2 graphs are defined analogously).
Given a memoryless strategy $\strab$ for player~2, we denote by 
$G^\strab$ the player-1 graph obtained by removing from all player-2 vertices
the edges that are not chosen by $\strab$.

\smallskip\noindent{\bf Multidimensional mean-payoff objectives.} 
For multidimensional mean-payoff objectives we will consider game graphs 
along with a weight function $w: E \to \Z^k$ that maps each edge to a vector
of integer weights.
We denote by $W$ the maximal absolute value of the weights.
For a finite path $\pat$, we denote by $w(\pat)$ the sum of the weight vectors 
of the edges in $\pat$ and $\Avg(\pat) = \frac{w(\pat)}{|\pat|}$, 
where $|\pat|$ is the length of $\pat$, denotes the average vector of the 
weights.
We denote by $\Avg_i(\pi)$ the projection of $\Avg(\pi)$ to the $i$-th 
dimension.
For an infinite path $\pat$, let $\rho_t$ denote the finite prefix of 
length $t$ of $\pat$; and we define 
$\LimInfAvg_i(\pat) =\lim\inf_{t \to\infty} \Avg_i(\rho_t)$ and analogously 
$\LimSupAvg_i(\pat)$ with $\lim\inf$ replaced by $\lim\sup$.
For an infinite path $\pi$, we denote by 
$\LimInfAvg(\pi)= (\LimInfAvg_1(\pi),\dots,\LimInfAvg_k(\pi))$ 
(resp. $\LimSupAvg(\pi) = (\LimSupAvg_1(\pi),\dots,\LimSupAvg_k(\pi))$)  
the limit-inf (resp. limit-sup) vector of the averages (long-run average or 
mean-payoff objectives).
The objective of player~1 we consider is to ensure that the 
mean-payoff is non-negative in every dimension, i.e., to ensure 
$\LimInfAvg(\pi) \geq \vec{0}$, where $\vec{0}$ denotes 
the vector of all zeros.

\begin{rem}
A mean-payoff objective is invariant to the shift operation, i.e., if in a 
dimension $i$, we require that the mean-payoff is at least $\nu_i$, then we 
simply subtract $\nu_i$ in the weight vector from every edge in the $i$-th
dimension and require the mean-payoff is at least $0$ in dimension~$i$.
Hence the comparison with $\vec{0}$ is without loss of generality. 
We will present all the results for $\LimInfAvg$ objectives and the
results for $\LimSupAvg$ objectives are simpler. 
Hence, in sequel we will write $\LimAvg$ for $\LimInfAvg$.
Moreover, all the results we will present would also hold if we replace 
the non-strict inequality comparison ($\geq \vec{0}$) with a strict
inequality ($> \vec{0}$).
\end{rem}

\smallskip\noindent{\bf Winning strategies.} A player-1 strategy $\straa$ is a 
winning strategy from a set $U$ of vertices, if for all player-2 strategies $\strab$ 
and all $v \in U$ we have $\LimAvg(\pat(v,\straa,\strab)) \geq \VEC{0}$. 
A player-2 strategy is a winning strategy from a set $U$ of vertices if for all 
player-1 strategies $\straa$ and for all $v \in U$ we have that the path 
$\pat(v,\straa,\strab)$ does not satisfy $\LimAvg(\pat(v,\straa,\strab)) \geq \VEC{0}$. 
The winning region for a player is the largest set $U$ such that the player 
has a winning strategy from $U$.

\subsection{Hyperplane separation algorithm}
In this subsection we will present our algorithm to decide the existence of 
a winning strategy for player~1 in finite-state multidimensional mean-payoff
games.

\smallskip\noindent{\bf Hyperplane separation technique.}
Our key insight is to search for a hyperplane $\mathcal{H}$ such that 
player~2 can ensure a mean-payoff vector that lies below $\mathcal{H}$.
Intuitively, we show that if such a hyperplane exists, then any point in space 
that is below $\mathcal{H}$ is negative in at least one dimension, and 
thus the multidimensional mean-payoff objective for player~1 is violated.
Conversely, we show that if for all hyperplanes $\mathcal{H}$ player~1 can 
achieve a mean-payoff vector that lies above $\mathcal{H}$, then player~1 
can ensure the multidimensional mean-payoff objective.
The technical argument relies on the fact that if we have an infinite sequence 
of unit vectors $\vec{b}_1,\vec{b}_2,\dots$ and $\vec{b}_\ell$ lies above the 
hyperplane that is normal to $\sum_{j=1}^{\ell-1} \vec{b}_j$, then 
$\lim\inf_{\ell\to\infty} \frac{1}{\ell}\cdot\sum_{j=1}^\ell \vec{b}_j = \vec{0}$.

\smallskip\noindent{\em Multiple dimensions to one dimension.}
Given a multidimensional weight function $w$ and a vector $\vec{\lambda}$, we 
denote by $w\cdot \vec{\lambda}$ the one-dimensional weight function that 
assigns every edge $e$ the weight value $w(e)^T\cdot \vec{\lambda}$, where 
$w(e)^T$ is the transpose of the weight vector $w(e)$.
We show that with the hyperplane technique we can reduce a game with 
multidimensional mean-payoff objective to the same game with 
a one-dimensional mean-payoff objective.
A vector $\vec{b}$ lies above a hyperplane $\mathcal{H}$ if $\vec{\lambda}$ is 
the normal vector of $\mathcal{H}$ and $\vec{b}^T \cdot \vec{\lambda} \geq 0$.
Hence, player~1 can achieve a mean-payoff vector that lies above $\mathcal{H}$ 
if and only if player~1 can ensure the one-dimensional mean-payoff objective
with weight function $w(e)\cdot \vec{\lambda}$.

\smallskip\noindent{\bf Examples.}
Consider the game graph $G_1$ (Figure~\ref{fig:G1}) where all vertices 
belong to player~1.
The weight function $w_1$ labels each edge with a two-dimensional weight
vector.
In $G_1$, player~1 can ensure all mean-payoff vectors that are 
convex combination of $(1,-2),(-2,1)$ and $(-1,-1)$ (see Figure~\ref{fig:G1Conv}).
By Figure~\ref{fig:G1Conv}, all the vectors reside below the hyperplane $y=-x$,
and consider the normal vector $\vec{\lambda} = (1,1)$ to the hyperplane $y=-x$.
All the cycles in $G_1$ with weight function $w_1 \cdot \vec{\lambda}$ 
(shown in Figure~\ref{fig:LambdaG1})
have negative weights. Therefore player~1 loses in the one-dimensional mean-payoff objective.
Consider the game graph $G_2$ (Figure~\ref{fig:G2}) with all player-1 vertices;
where player~1 can achieve any mean-payoff vector that is a convex combination of 
$(2,-1),(-1,2)$ and $(-2,-1)$ (see Figure~\ref{fig:G2Conv}).
By Figure~\ref{fig:G2Conv}, every two-dimensional hyperplane that passes 
through the origin intersects with the feasible region. 
Thus, no separating hyperplane exists.
\begin{figure}
\begin{minipage}[t]{0.45\linewidth}
	 \begin{center} 
    \begin{gpicture}(25, 10)(0,-5)
    \thinlines

    \node[Nw=4.0,Nh=4.0](A1)(0,0){$v_0$}
    \drawloop[loopangle=180,loopdiam=4](A1){$(1,-2)$}
    \node[Nw=4.0,Nh=4.0](A2)(25,0){$v_2$}
    \drawloop[loopangle=0,loopdiam=4](A2){$(-2,1)$}

    \node[Nw=4.0,Nh=4.0](A3)(12.5,0){$v_1$}

    \drawedge(A1,A3){$(0,0)$}
    \drawedge(A3,A2){$(0,0)$}
    \gasset{Nw=5,Nh=5,Nmr=2.5,curvedepth=6}
    \drawedge(A2,A1){$(-3,-3)$}
    \end{gpicture}
\end{center}
  \caption{Game graph $G_1$.}
\label{fig:G1}
\end{minipage}
\hfill
\begin{minipage}[t]{0.45\linewidth}	 
\begin{center} 
    \begin{gpicture}(25, 10)(0,-5)
    \thinlines

    \node[Nw=4.0,Nh=4.0](A1)(0,0){$v_0$}
    \drawloop[loopangle=180,loopdiam=4](A1){$(2,-1)$}
    \node[Nw=4.0,Nh=4.0](A2)(25,0){$v_2$}
    \drawloop[loopangle=0,loopdiam=4](A2){$(-1,2)$}

    \node[Nw=4.0,Nh=4.0](A3)(12.5,0){$v_1$}

    \drawedge(A1,A3){$(0,0)$}
    \drawedge(A3,A2){$(0,0)$}
    \gasset{Nw=5,Nh=5,Nmr=2.5,curvedepth=6}
    \drawedge(A2,A1){$(-6,-3)$}
    \end{gpicture}
	 \end{center}
  \caption{Game graph $G_2$.}
\label{fig:G2}
\end{minipage}
\end{figure}

\begin{figure}
\begin{minipage}[t]{0.45\linewidth}	 
\setlength\unitlength{1cm}
\begin{center} 
\begin{pspicture}(1,1)
{\color{gray}\polygon*(0.5,0.5)(0.75,0)(0,0)(0,0.75)}
\psline[linewidth=0.5pt](0.5,0.5)(1.25,0.5)
\psline[linewidth=0.5pt](0.5,0.5)(-0.25,0.5)
\psline[linewidth=0.5pt](0.5,0.5)(0.5,-0.25)
\psline[linewidth=0.5pt](0.5,0.5)(0.5,1.25)

\psline{->}(0.5,0.5)(0,0)
\psline{->}(0.5,0.5)(0.75,0)
\psline{->}(0.5,0.5)(0,0.75)
\end{pspicture}
\end{center} 
  \caption{Feasible vectors for $G_1$.}
\label{fig:G1Conv}
\end{minipage}
\hfill
\begin{minipage}[t]{0.45\linewidth}	 
\setlength\unitlength{0.25cm}
\begin{center} 
\begin{pspicture}(1,1)
{\color{gray}\polygon*(0,0.25)(1,0.25)(0.25,1)}

\psline[linewidth=0.5pt](0.5,0.5)(1.25,0.5)
\psline[linewidth=0.5pt](0.5,0.5)(-0.25,0.5)
\psline[linewidth=0.5pt](0.5,0.5)(0.5,-0.25)
\psline[linewidth=0.5pt](0.5,0.5)(0.5,1.25)

\psline{->}(0.5,0.5)(0,0.25)
\psline{->}(0.5,0.5)(1,0.25)
\psline{->}(0.5,0.5)(0.25,1)
\end{pspicture}
\end{center} 
  \caption{Feasible vectors for $G_2$.}
\label{fig:G2Conv}
\end{minipage}
\end{figure}

\begin{figure}
	 \begin{center} 
    \begin{picture}(25, 10)(0,-5)
    \thinlines

    \node[Nw=4.0,Nh=4.0](A1)(0,0){$v_0$}
    \drawloop[loopangle=180,loopdiam=4](A1){$-1$}
    \node[Nw=4.0,Nh=4.0](A2)(25,0){$v_2$}
    \drawloop[loopangle=0,loopdiam=4](A2){$-1$}

    \node[Nw=4.0,Nh=4.0](A3)(12.5,0){$v_1$}

    \drawedge(A1,A3){$0$}
    \drawedge(A3,A2){$0$}
    \gasset{Nw=5,Nh=5,Nmr=2.5,curvedepth=6}
    \drawedge(A2,A1){$-6$}
    \end{picture}
	 \end{center}
  \caption{Game graph $G_1$ with weight function $\vec{\lambda}\cdot w_1$ for $\vec{\lambda} = (1,1)$.}
\label{fig:LambdaG1}
\end{figure}
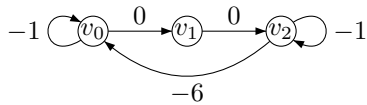

\smallskip\noindent{\em Basic lemmas and assumptions.}
We now prove two lemmas to formalize the intuition related to reduction 
to one-dimensional mean-payoff games.
Lemma~\ref{lem:ReductionToOneDimEasyDirection} requires two assumptions, 
which we later show (in Lemma~\ref{lem:WithoutAssump}) how to deal with.
The assumptions are as follows:
(1)~The first assumption (we refer as Assumption~1) is that every outgoing edge 
of player-2 vertices is to a player-1 vertex; 
formally, $E \cap (V_2 \times V) \subseteq E \cap (V_2 \times V_1)$.
(2)~The second assumption (we refer as Assumption~2) is that every player-1 vertex has $k$ self-loop edges 
$e_1,\dots,e_k$ such that $w_i(e_j) = 0$ if $i\neq j$ and $w_i(e_i) = -1$.
Let us denote by $\Win^2$ the player-2 winning region in the multidimensional
mean-payoff game with weight function $w$, and 
by $\Win^2_{\vec{\lambda}}$ the player-2 winning region in the one-dimensional mean-payoff game 
with the weight function $w\cdot \vec{\lambda}$. 
The next lemma shows that if $\Win^2_{\vec{\lambda}} \neq \emptyset$, then $\Win^2\neq \emptyset$;
i.e., presents a sufficient condition for the non-emptiness of $\Win^2$.

\begin{lem}\label{lem:ReductionToOneDimEasyDirection}
Given a game graph $G$ that satisfies Assumption~1 and Assumption~2, and a 
multidimensional mean-payoff objective with weight function $w$,
for every $\vec{\lambda}\in\R^k$ we have $\Win^2_{\vec{\lambda}} \subseteq \Win^2$; 
(hence, if $\Win^2_{\vec{\lambda}} \neq \emptyset$, then $\Win^2\neq \emptyset$).
\end{lem}
\begin{proof}
Let $\strab$ be a player-2 winning strategy in $G$ from an initial vertex 
$v_0$ (i.e., winning strategy from the set $\set{v_0}$) for the one-dimensional 
mean-payoff objective with weight function $w\cdot\vec{\lambda}$.
We first observe that we must have $\vec{\lambda} \in (0,\infty)^k$; otherwise 
if $\lambda_i \in (-\infty,0]$ then by Assumption~1 the weight of the $i$-th self-loop of a 
player-1 vertex would be non-negative, and player~1 can ensure the 
mean-payoff objective from all vertices (by Assumption~2 all plays arrive to a player-1 vertex
within one step), contradicting $v_0$ is winning for player~2.
We claim that $\strab$ is also a player-2 winning strategy with respect to the multidimensional mean-payoff 
objective.
Indeed, let $\rho$ be a play that is consistent with $\strab$.
Since $\strab$ is a player-2 winning strategy for the mean-payoff objective with weight 
function $w\cdot\vec{\lambda}$, it follows that there exists a constant $c > 0$ such 
that there are infinitely many prefixes of $\rho$ with average weight (according to $w\cdot\vec{\lambda}$) 
at most $-c$.
Let $\lambda_{\min}= \min \set{\vec{\lambda}_i \mid 1 \leq i \leq k}$ be the minimum 
value of $\vec{\lambda}$ among its dimension.
Since $\vec{\lambda} \in (0,\infty)^k$, it follows that $\lambda_{\min} >0$.
Since there are finitely many dimensions 
there must be a dimension $i$ for which there are infinitely many prefixes of $\rho$ 
with average weight at most $-\frac{c \cdot \lambda_{\min}}{k} <0$ 
in dimension $i$. 
Hence, the mean-payoff value of dimension $i$ is negative, and 
thus the multidimensional mean-payoff objective is violated.
Hence $\strab$ is a player-2 winning strategy from $v_0$ against 
the multidimensional mean-payoff objective.
\end{proof}

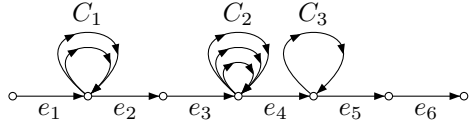
\begin{figure}[!tb]
\begin{center} 
  \begin{picture}(132,61)(-46,-18)
	\node[Nw=1,Nh=1](v1)(0,0){}
	\node[Nw=1,Nh=1](v2)(10,0){}
	\node[Nw=1,Nh=1](v3)(20,0){}
	\node[Nw=1,Nh=1](v4)(30,0){}
	\node[Nw=1,Nh=1](v5)(40,0){}
	\node[Nw=1,Nh=1](v6)(50,0){}
	\node[Nw=1,Nh=1](v7)(60,0){}

	\drawedge[ELside=r](v1,v2){$e_1$}
	\drawedge[ELside=r](v2,v3){$e_2$}
	\drawedge[ELside=r](v3,v4){$e_3$}
	\drawedge[ELside=r](v4,v5){$e_4$}
	\drawedge[ELside=r](v5,v6){$e_5$}
	\drawedge[ELside=r](v6,v7){$e_6$}

	\drawloop[loopdiam=6,loopangle=90,AHnb=3,AHdist=4.71](v2){}
	\drawloop[loopdiam=8,loopangle=90,AHnb=3,AHdist=6.28](v2){$C_1$}

	\drawloop[loopdiam=4,loopangle=90,AHnb=3,AHdist=3](v4){}
	\drawloop[loopdiam=6,loopangle=90,AHnb=3,AHdist=4.71](v4){}
	\drawloop[loopdiam=8,loopangle=90,AHnb=3,AHdist=6.28](v4){$C_2$}

	\drawloop[loopdiam=8,loopangle=90,AHnb=3,AHdist=6.28](v5){$C_3$}

  \end{picture}
\caption{The path segment $\rho_i$ is decomposed into the cycles 
(possibly repeated) as $\rho_i^2 = C_1 \cdot C_1 \cdot C_2\cdot C_2\cdot 
C_2\cdot C_3$; and the acyclic part $\rho_i^1 = e_1 \cdot e_2 \cdot e_3\cdot e_4\cdot e_5\cdot e_6$.}
\end{center}\label{fig:hard-dir}
\end{figure}

We now present a lemma that will complement Lemma~\ref{lem:ReductionToOneDimEasyDirection},
and the following lemma does not require Assumption~1 or Assumption~2.

\begin{lem}\label{lem:ReductionToOneDimHardDirection}
Given a game graph $G$ and a multidimensional mean-payoff objective with weight function $w$,
if for all $\vec{\lambda} \in \R^k$ we have $\Win^2_{\vec{\lambda}} = \emptyset$, 
then we have $\Win^2=\emptyset$.
\end{lem}
\begin{proof}
Since $\Win^2_{\vec{\lambda}}=\emptyset$ for every $\vec{\lambda} \in \R^k$, it follows by the determinacy 
of one-dimensional mean-payoff games~\cite{EM79} that for all $\vec{\lambda}\in \R^k$, player~1 can 
ensure the one-dimensional mean-payoff objective with weight function $w\cdot \vec{\lambda}$ in $G$ 
(from all initial vertices). 
We now present an explicit construction of a player-1 winning strategy for the multidimensional mean-payoff
objective in $G$.
For a vector $\vec{\lambda}\in \R^k$, let $\straa_{\vec{\lambda}}$ be a 
memoryless player-1 winning strategy in $G$ from all vertices for the 
one-dimensional mean-payoff game with weight function $w\cdot\vec{\lambda}$ 
(note that uniform memoryless winning strategies that ensure winning from all
vertices in the winning region exist in one-dimensional mean-payoff games by 
the results of~\cite{EM79}).
We construct a player-1 winning strategy $\straa$ for the multidimensional objective 
in the following way:
\begin{itemize}
\item Initially, set $\VEC{b}_0 := (1,1,\dots,1)$.
\item For $i=1,2,\dots,\infty$, in iteration $i$ play as follows:
\begin{itemize}
\item Set $\vec{\lambda}_{b_i} := -\VEC{b}_{i-1}$. 
In $\straa$, player~1 plays according to $\straa_{\vec{\lambda}_{b_i}}$ for $i$ rounds.
\item Let $\rho_i$ be the play suffix that was formed in the last $i$ rounds (or steps) of the play.
From $\rho_i$ we obtain the part of $\rho_i$ that consists of cycles (that are possibly repeated)
and denote the part as $\rho_i^2$; and an acyclic part $\rho_i^1$ of length at most $n$.
Informally, $\rho_i^2$ consists of cycles $C$ that appear in $\rho_i$, and if cycle 
$C$ is repeated $j$ times in $\rho_i$ then it is included $j$ times in $\rho_i^2$;
see Figure~\ref{fig:hard-dir} for an illustration.
\item Set $\VEC{b}_i := \VEC{b}_{i-1} + w(\rho_i^2)$; and proceed to the next iteration.
\end{itemize}
\end{itemize}
In order to prove that $\straa$ is a winning strategy, it is enough to prove that for every play $\rho$ 
that is consistent with $\straa$, the Euclidean norm of the average weight vector tends to zero 
as the length of the play tends to infinity.

We first compute the Euclidean norm of $\VEC{b}_i$.
For this purpose we observe that $\straa_{\vec{\lambda}_{b_i}}$ is a memoryless winning strategy 
for the one-dimensional mean-payoff game with weight function $w\cdot\vec{\lambda}_{b_i}$;
and hence it follows that for every cycle $C$ in the graph $G^{\straa_{\vec{\lambda}_{b_i}}}$ 
the sum of the weights of $C$ according to $w \cdot  \vec{\lambda}_{b_i}$ is non-negative.
Since $\rho_i^2$ is composed of cyclic paths, we must have 
$w(\rho_i^2)^T \cdot \vec{\lambda}_{b_i} \geq 0$; and hence, 
we have $w(\rho_i^2)^T \cdot \VEC{b}_{i-1} \leq 0$.
Thus we get that
\[|\VEC{b}_{i}|=|\VEC{b}_{i-1} + w(\rho_i^2)| = \sqrt{|\VEC{b}_{i-1}|^2+2  \cdot w(\rho_i^2)^T \cdot \VEC{b}_{i-1} + |w(\rho_i^2)|^2} 
\leq \sqrt{|\VEC{b}_{i-1}|^2+ |w(\rho_i^2)|^2}\]
Since $W$ is the maximal absolute value of the weights, it follows that 
$W \cdot \sqrt{k}$ is a bound on the Euclidean norm of any average weight vector.
Since the length of $\rho_i^2$ is at most $i$ (it was a part of the suffix of last 
$i$ rounds)
we get that
\[|\VEC{b}_{i}|\leq \sqrt{|\VEC{b}_{i-1}|^2 + k \cdot W^2 \cdot i^2}.\]
By a simple induction we obtain that 
$|\VEC{b}_i| \leq  \sqrt{k \cdot W^2 \cdot \sum_{j=1}^i j^2}$.
Thus we have 
\[|\VEC{b}_i| \leq \sqrt{k \cdot W^2 \cdot \sum_{j=1}^i j^2} \leq \sqrt{k \cdot W^2 \cdot i^3}\enspace.\]

We are now ready to compute the the Euclidean norm of the play after the $i$-th iteration.
We denote the weight vector after the $i$-th iteration by $\VEC{x}_i$ and observe that
\[ \VEC{x}_i = \VEC{b}_i + \sum_{j=1}^i w(\rho_j^1) \enspace.\]
By the Triangle inequality we get that 
\[ |\VEC{x}_i| \leq |\VEC{b}_i| + \sum_{j=1}^i |w(\rho_j^1)| \enspace.\]
Since the length of $\rho_i^1$ is at most $n$ and by the bound we obtained over $\VEC{b}_i$ we get that 
\[ 
|\VEC{x}_i| \leq \sqrt{k\cdot W^2 \cdot i^3} + i\cdot n\cdot W \cdot \sqrt{k}
\]
For a position $j$ of the play between iteration $i$ and iteration $i+1$, 
let us denote by $\vec{y}_j$ the weight vector after the play prefix
at position $j$.
Since there are  $i$ steps played in iteration $i$ 
we have $|\vec{y}_j| \leq |\vec{x}_i| + i \cdot W \cdot \sqrt{k}$.
Finally, since after the $(i-1)$-th iteration $\sum_{t=1}^{i-1} t= i\cdot(i-1)/2$ rounds were played, 
we get that the Euclidean norm of the average weight vector, namely,
$|\frac{\VEC{y}_j}{j}| \leq \frac{|\VEC{y}_j|}{i\cdot(i-1)/2}$, 
tends to zero as $i$ tends to infinity.
Formally we have 
\[
\lim_{j \to \infty} \frac{|\vec{y}_j|}{j} \leq 
\lim_{i \to \infty} 
\frac{\sqrt{k\cdot W^2 \cdot i^3} + i\cdot n \cdot W \cdot \sqrt{k} + i \cdot W \cdot \sqrt{k}}{i\cdot(i-1)/2} =0
\]
It follows that the limit average of the weight vectors is zero
and hence the desired result follows. 
\end{proof}

Lemma~\ref{lem:ReductionToOneDimEasyDirection} and Lemma~\ref{lem:ReductionToOneDimHardDirection}
suggest that in order to check if player-2 winning region is non-empty in a 
multidimensional mean-payoff game it suffices to go over all (uncountably many) $\vec{\lambda}\in\R^k$ 
and check whether player-2 winning region is non-empty in the one-dimensional mean-payoff game 
with weight function $w\cdot \vec{\lambda}$. 
The next lemma shows that we need to consider only finitely many vectors; and 
we first introduce some notations that we will use.

\smallskip\noindent{\em Notations.}
For the rest of this section, we denote $M = (k\cdot n\cdot W)^{k+1}$, where $W$ is 
the maximal absolute value of the weight function.
For a positive integer $\ell$, we will denote by $\Z^{\pm}_\ell=\set{i \mid -\ell \leq i \leq \ell}$
(resp.  $\Z^+_\ell =\set{i \mid 1 \leq i \leq \ell}$) the set of integers (resp. positive integers) 
from $-\ell$ to $\ell$.

\begin{lem}\label{lem:FinitelyManyLambda}
Let $G$ be a game graph with a multidimensional mean-payoff objective with a weight function $w$.
There exists $\vec{\lambda}_0\in \R^k$ for which player-2 winning region is non-empty in $G$ 
for the one-dimensional mean-payoff objective with weight function $w\cdot \vec{\lambda}_0$ 
if and only if there exists $\vec{\lambda}\in (\Z^{\pm}_M)^k$ 
such that the player-2 winning region is non-empty in $G$ for the one-dimensional mean-payoff 
objective with weight function $w\cdot \vec{\lambda}$.
\end{lem}
\begin{proof}
Suppose that player~2 has a memoryless winning strategy $\strab$ in $G$ from an 
initial vertex $v_0$ for the one-dimensional mean-payoff objective with weight 
function $w\cdot \vec{\lambda}_0$.
Let $C_1,\dots,C_m$ be the simple cycles that are reachable from $v_0$ in the graph $G^\strab$.
Since $\strab$ is a player-2 winning strategy it follows that $w(C_i)^T \cdot \vec{\lambda}_0  < 0$ 
for every $i\in\{1,\dots,m\}$.
We note that for all $1 \leq i \leq m$ we have $w(C_i)\in (\Z^{\pm}_{n \cdot W})^k$
(since $C_i$ is a simple cycle, in every dimension the sum of the weights is 
between $-n \cdot W$ and $n \cdot W$).
Then by~\cite[Lemma~2, items c and d]{Papa81} it follows that there is a vector of integers 
$\vec{\lambda}$ such that $w(C_i)^T\cdot \vec{\lambda} \leq -1 < 0$, for all $1 \leq i \leq m$; 
and $\vec{\lambda} \in (\Z^{\pm}_M)^{k}$.
Since all the reachable cycles from $v_0$ in $G^{\strab}$ are negative according to $w\cdot \vec{\lambda}$, 
we get that $\strab$ is a winning strategy for the one-dimensional mean-payoff game with weight function 
$w\cdot \vec{\lambda}$; and 
hence the proof for the direction from left to right follows.
The proof for the converse direction is trivial.
\end{proof}

The next lemma removes the two assumptions of 
Lemma~\ref{lem:ReductionToOneDimEasyDirection}.

\begin{lem}\label{lem:WithoutAssump}
Let $G$ be a game graph with a multidimensional mean-payoff objective with a weight function $w$.
The following assertions hold:
(1)~$
\bigcup_{\vec{\lambda}\in (\Z_M^{+})^k} \Win^2_{\vec{\lambda}} 
\subseteq \Win^2$.
(2)~If 
$
\bigcup_{\vec{\lambda}\in (\Z_M^{+})^k} \Win^2_{\vec{\lambda}}=\emptyset$, 
then $\Win^2 = \emptyset$.
\end{lem}
\begin{proof}
We first show how to construct a game graph $\wh{G}$ from $G$ that satisfies the 
two assumptions (Assumption~1 and Assumption~2) and has the same winning regions 
(for the multidimensional objective) as in $G$.
\begin{enumerate}
\item \emph{(Assumption~1).}
Given any game graph $G$ there exists a linear transformation to satisfy 
Assumption~1 by simply adding a dummy vertex for every outgoing edge
of a player~2 vertex (i.e., for every edge $e=(u,v)$ with $u, v \in V_2$,
we add a vertex $e$, edges $(u,e)$ with weight $w(e)$ and $(e,v)$ with 
weight $\VEC{0}$, and $e$ is a player-1 vertex with a single outgoing 
edge).

\item \emph{(Assumption~2).}
First, note that adding several self-loop edges creates a multi-graph, but 
a dummy player-2 vertex can be put for every such edge to ensure that 
we do not have a multi-graph. 
Second we observe that adding the self-loop edges of Assumption~2 
do not affect winning 
for player~1, as if there is a winning strategy for player~1, then 
there is one that never chooses the self-loop edges of Assumption~2 
because the self-loop edges are non-positive in every dimension and negative
in one dimension. 
\end{enumerate}
For a game graph $G$ we denote by $\wh{G}$ the graph that is formed by the 
transformations above.
We now establish the following claim:

\smallskip\noindent{\em Claim.} The following two properties hold for the game 
graph $\wh{G}$:
(i)~if a vector $\vec{\lambda}$ is non-positive in (at least) one dimension, then player-2 winning region 
in $\wh{G}$ for the one-dimensional mean-payoff objective with weight function $w\cdot\vec{\lambda}$ is empty; and
(ii)~if a vector $\vec{\lambda}$ is positive in all dimensions, then player-2 winning region in $G$ and in $\wh{G}$ 
is the same for the one-dimensional mean-payoff objective with weight function $w\cdot\vec{\lambda}$.
The first item of the claim holds due to the self-loops of Assumption~2, and Assumption~1 ensures that 
a player-1 vertex is reached within two steps (the same reasoning as used in Lemma~\ref{lem:ReductionToOneDimEasyDirection}).
The second item of the claim holds because the weight of any simple cycle in $G$ is the same as in $\wh{G}$, 
and the weight of Assumption~2 self-loops are non-positive in every dimension and negative in one dimension 
(since $\vec{\lambda}$ is positive in all dimensions). 
Hence, a memoryless winning strategy in $G$ is also winning in $\wh{G}$ and vice-versa.

We now prove the two assertions of the lemma.
\begin{enumerate}
\item \emph{(First assertion).}
Consider that in $G$ we have $v\in\Win^2_{\vec{\lambda}}$, for some vertex $v$ 
and a vector $\vec{\lambda}\in(\Z_M^{+})^k$.
Then by the second item of the claim we get that $v\in\Win^2_{\vec{\lambda}}$ 
also in $\wh{G}$, and then by Lemma~\ref{lem:ReductionToOneDimEasyDirection} 
we get that $v\in\Win^2$ (in $\wh{G}$).
Finally, by the definition of the transformations, we get that player~2 is 
winning from $v$ for the multidimensional mean-payoff objective in $\wh{G}$ 
if and only if player~2 is winning from $v$ for the multidimensional 
mean-payoff objective in $G$.
Thus $v\in\Win^2$ in $G$ and the first assertion follows.

\item \emph{(Second assertion).}
For the second assertion consider that $\Win^2\neq \emptyset$ (in $G$) 
and we show that for some $\vec{\lambda}\in (\Z_M^{+})^k$ we have 
$\Win^2_{\vec{\lambda}} \neq \emptyset$ (in $G$).
Suppose that $v\in\Win^2$ for some vertex $v$ in $G$.
Then by the definition of the transformation we have that $v \in\Win^2$ 
also in $\wh{G}$.
By Lemma~\ref{lem:ReductionToOneDimHardDirection} and 
Lemma~\ref{lem:FinitelyManyLambda} it follows that there is 
$\vec{\lambda}\in(\Z_M^{\pm})^k$ such that $v\in\Win^2_{\vec{\lambda}}$ 
(in $\wh{G})$. 
By the first item of the claim we get that $\vec{\lambda}\in(\Z_M^{+})^k$.
Finally, by the second item of the claim, and since 
$\vec{\lambda}\in(\Z_M^{+})^k$, we get that $v\in\Win^2_{\vec{\lambda}}$ also 
in  $G$, and thus the second assertion follows.
\end{enumerate}
The desired result follows.
\end{proof}

To use the result of Lemma~\ref{lem:WithoutAssump} iteratively
to solve finite-state games with multidimensional mean-payoff objectives,
we need the notion of \emph{attractors}.
For a set $U$ of vertices,  $\AttrTwo{U}$ is defined inductively as follows:
$U_0=U$ and for all $i\geq 0$ we have 
$U_{i+1} = U_i \cup \Set{v \in V_1 \mid \Out(v) \subseteq U_i} \cup 
\Set{v \in V_2 \mid \Out(v) \cap U_i \neq \emptyset}$, 
and  $\AttrTwo{U}= \bigcup_{i\geq 0} U_i$.
Intuitively, from $U_{i+1}$ player~2 can ensure to reach $U_i$ in one step 
against all strategies of player~1, and thus  $\AttrTwo{U}$ is the set of 
vertices such that player~2 can ensure to reach $U$ against all strategies of 
player~1 in finitely many steps.
The set $\AttrTwo{U}$ can be computed in linear time~\cite{Immerman81,Bee80}.
Observe that if $G$ is a game graph, then for all $U$, 
the game graph induced by the set $V \setminus  \AttrTwo{U}$ 
is also a game graph (i.e., all vertices in $V \setminus \AttrTwo{U}$ 
have outgoing edges in $V \setminus \AttrTwo{U}$).
The following lemma shows that in multidimensional mean-payoff games,
if $U$ is a set of vertices such that player~2 has a winning strategy
from every vertex in $U$, then player~2 has a winning strategy 
from all vertices in $\AttrTwo{U}$, and we can recurse in the game 
graph after removal of $\AttrTwo{U}$.

\begin{lem}\label{lem:iter}
Consider a multidimensional mean-payoff game $G$ with weight function $w$.
Let $U$ be a set of vertices such that from all vertices in $U$ there is a 
winning strategy for player~2.
Then the following assertions hold:
(1)~From all vertices in  $\AttrTwo{U}$ there is a winning strategy for 
player~2.
(2)~Let $Z$ be the set of vertices in the game graph induced after removal of 
$\AttrTwo{U}$ such that from all vertices in $Z$ player~2 has a winning 
strategy in the remaining game graph. 
Then  from all vertices in $Z$, player~2 has a winning strategy in the original 
game graph.
\end{lem}
\begin{proof}
The proof of the first item is as follows: from vertices in  $\AttrTwo{U}$ 
first consider a strategy to ensure to reach $U$ (within finitely many 
steps), and once $U$ is reached switch to a winning strategy from vertices
in $U$.
The proof of second item is as follows: fix a winning strategy in 
the remaining game graph for vertices in $Z$ and a winning strategy from  $\AttrTwo{U}$
for player~2. Consider any counter strategy for player~1. 
If  $\AttrTwo{U}$ is ever reached, then the winning strategy from  $\AttrTwo{U}$
ensures winning for player~2, and otherwise the winning strategy of 
the remaining game graph ensures winning.
\end{proof}

\smallskip\noindent{\bf Algorithm.} We now present our iterative algorithm 
that is based on Lemma~\ref{lem:WithoutAssump} and Lemma~\ref{lem:iter}.
In the current iteration $i$ of the game graph execute the following steps:
sequentially iterate over vectors $\vec{\lambda}\in(\Z^{+}_M)^k$;
and if for some $\vec{\lambda}$ we obtain a non-empty set $U$ of winning vertices 
for player~2 for the one-dimensional mean-payoff objective with weight function 
$w \cdot \vec{\lambda}$
in the current game graph, remove $\AttrTwo{U}$ from the current game graph and 
proceed to iteration $i+1$.
Otherwise if for all $\vec{\lambda}\in(\Z^{+}_M)^k$,
player~1 wins from all vertices for the one-dimensional mean-payoff objective
with weight function $w\cdot \vec{\lambda}$, then the set of current vertices is the set of 
winning vertices for player~1.
The correctness of the algorithm follows from Lemma~\ref{lem:WithoutAssump}
and Lemma~\ref{lem:iter}.

\smallskip\noindent{\bf Complexity.}
The algorithm has at most $n$ iterations, and each iteration solves at most 
$O(M^{k})$ one-dimensional mean-payoff games.
Thus the iterative algorithm requires to solve 
$O(n \cdot M^{k})$ one-dimensional mean-payoff 
games with $m$ edges, $n$ vertices, and the maximal weight is at most 
$k \cdot W\cdot M$.
Since one-dimensional mean-payoff games with $n$ vertices, $m$ edges, and 
maximal weight $W$ can be solved in time $O(n \cdot m \cdot W)$~\cite{BCDGR11},
we obtain the following result.


\begin{thm}\label{thm1}
The set of winning vertices for player~1 in a multidimensional mean-payoff 
game with $n$ vertices, $m$ edges, $k$-dimensions, and maximal absolute weight 
$W$ can be computed in time 
$
O(n^2 \cdot m \cdot k \cdot W \cdot  (k \cdot n \cdot W)^{k^2 + 2\cdot k +1})$.
\end{thm}


\subsection{Hardness for fixed parameter tractability}
In this subsection we will reduce finite-state parity games to 
finite-state multidimensional mean-payoff games with weights bounded 
linearly by the number of vertices.
Note that our reduction is different from the standard reduction of parity games 
to one-dimensional mean-payoff games where exponential weights
are necessary~\cite{Jur98}.
We start with the definition of parity games.

\smallskip\noindent{\bf Parity games.}
A parity game consists of a finite-state game graph $G$ along with a priority 
function $p:E\to\RangeSet{1}{k}$ that maps every edge to a natural number 
(the priority).
The objective of player 1 is to ensure that the \emph{minimal} priority 
that occurs infinitely often in a play is \emph{even}, and 
the goal of player~2 is the complement.
The memoryless determinacy of parity games shows that for both 
players if there is a winning strategy, then there is a memoryless
winning strategy~\cite{EJ91}.

\Heading{The reduction.}
Given a game graph $G$ with priority function $p$ we construct a multidimensional 
mean-payoff objective with weight function $w$ of $k$ dimensions on $G$ as
follows:  for every $i\in\RangeSet{1}{k}$ we assign $w_i(e)$ as follows:
\begin{itemize}
\item $0$ if $p(e) > i$;
\item $-1$ if $p(e) \leq i$ and $p(e)$ is odd; and 
\item $n$ if $p(e) \leq i$ and $p(e)$ is even.
\end{itemize}

\begin{lem}\label{lem:Player1WinsParityThenWinsMP}
From a vertex $v$, if player~1 wins the parity game, then she also wins the multidimensional mean-payoff game.
\end{lem}
\begin{proof}
If player 1 is the winner in the parity game from $v$, then by memoryless determinacy of parity games 
there is memoryless winning strategy $\straa$.
Since $\straa$ is winning in the parity game, then every simple cycle $C$ reachable from $v$ in $G^\straa$ is even
(i.e., the minimum priority of $C$ is even).
Given a cycle $C$ with minimum priority $i$ which is even we have 
(i)~for $j<i$: $w_j(C)=0$; 
and (ii)~for $j\geq i$ there is at least one state with weight $n$, and the sum of all 
other weights is at least $-(n-1)$ (since there are at most $n$ edges of which one has
weight $n$, and in the worst case all the remaining $n-1$ edges have weight $-1$);
and hence $w_j(C) \geq 0$.
Hence by the construction of the weight function it follows that the weight vector of $C$ 
is non-negative (in every dimension).
Thus $\straa$ is a winning strategy for the multidimensional mean-payoff objective.
\end{proof}

\begin{lem}\label{lem:Player2WinsParityThenWinsMP}
From a vertex $v$, if player~2 wins the parity game, then she also wins the multidimensional mean-payoff game.
\end{lem}
\begin{proof}
If player 2 is the winner in the parity game from $v$, then by memoryless determinacy she has a memoryless 
winning strategy $\strab$.
We claim that $\strab$ is a winning strategy for player~2 in the multidimensional mean-payoff game.
For this purpose we first show that $\strab$ is a winning strategy in the one-dimensional mean-payoff game with 
weight function $w\cdot\vec{\lambda}$, where $\ell=n^2$ and  
\[\vec{\lambda} = (\ell^{k-1}, \ \ell^{k-2}, \ \dots, \ \ell^{k-i}, \ \dots, \ \ell^{0})
\]
Let $C$ be a simple cycle reachable from $v$ in the player-1 graph $G^\strab$.
Let $i$ be the minimal priority that occurs in $C$, and since $\strab$ is winning for player~2,
it follows that $i$ is odd.
By the construction of the weight function we get that (i)~$w_i(C)\leq -1$; 
(ii)~for $j>i$: $w_j(C)\leq n^2-1=\ell-1$ (at least one edge has negative 
weight, and all other edges have weight at most $n$); and 
(iii)~for $j<i$: $w_j(C)=0$.
Hence we get that
\[
w(C)^T\cdot \vec{\lambda} \leq -\ell^{k-i} + (\ell-1) \cdot \sum_{j=i+1}^k \ell^{k-j}
\leq \ell^{k-i} + (\ell-1)\cdot \ell^{k-i-1} < 0
\]
Hence, we get that every cycle reachable from $v$ in $G^\strab$ is negative according to $w \cdot \vec{\lambda}$; and hence 
$\strab$ is a winning strategy in the one-dimensional mean-payoff game for weight function $w \cdot \vec{\lambda}$.
By Lemma~\ref{lem:WithoutAssump} it follows that player~2 also wins in the multidimensional mean-payoff game
from $v$.
\end{proof}

\begin{thm}\label{thm:finite-reduction}
Let $G$ be a game graph with a parity objective defined by a priority function 
of $k$-priorities. We can construct in linear time a $k$-dimensional weight 
function $w$, with maximal weight $W$ bounded by $n$, such that a vertex is 
winning for player~1 in the parity game iff the vertex is winning for 
player~1 in the multidimensional mean-payoff game.  
\end{thm}

\begin{rem}
There exists a deterministic sub-exponential time algorithm for parity 
games~\cite{JPZ08} and also algorithms that run in time $O(n^{k/3}\cdot m)$~\cite{Schewe07}; 
however obtaining a fixed parameter tractable algorithm for 
parity games that runs in time $O(f(k) \cdot \text{poly}(n,m))$ 
for any function $f$ (exponential or double exponential) 
is a long-standing open problem.
Our reduction (Theorem~\ref{thm:finite-reduction}) shows that obtaining 
a fixed parameter tractable algorithm for multidimensional mean-payoff
games that runs in time $O(f(k) \cdot \text{poly}(n,m,W))$ is not possible
without first solving the fixed parameter tractability of parity games. 
We also point out that the hardness result does not hold for multidimensional 
$\LimSupAvg$-objectives, as if the weights are fixed, the problem can be 
solved in polynomial time~\cite{VR11}.
\end{rem}

%% file: wps-multi.tex
\newcommand{\PumpSet}{\mathsf{PPS}}
\newcommand{\PumpCone}{\mathsf{PCone}}
\newcommand{\PumpMatrix}{\mathsf{PumpMat}}
\newcommand{\Cone}{\mathsf{cone}}
\newcommand{\Rkplus}{\mathbb{R}^k_{+}}
\newcommand{\Rplus}{\mathbb{R}_{+}}
\newcommand{\VECT}{\mathsf{VECT}}

\section{Pushdown Graphs with Multidimensional Mean-payoff Objectives}
In this section we consider pushdown graphs (or pushdown systems) with multidimensional mean-payoff objectives, and we give an algorithm that determines if there exists a path that satisfies a multidimensional objective.
The algorithm we propose runs in polynomial time even for arbitrary number of dimensions and for arbitrary weight function.
As in the previous section, we use the hyperplane separation technique to reduce the problem into a one-dimensional pushdown graphs, and a polynomial solution for the latter is known~\cite{CV12}.

\smallskip\noindent{\bf Key obstacles and overview of the solution.}
We first describe the key obstacles for the polynomial time algorithm to 
solve pushdown graphs with multidimensional mean-payoff objectives (as 
compared to finite-state graphs and finite-state games).
For pushdown graphs we need to overcome the next three main obstacles:
(a)~The mean-payoff value of a finite-state graph is uniquely determined by 
the weights of the simple cycles of the graph. However, for pushdown graphs 
it is also possible to \emph{pump} special types of acyclic paths. 
Hence, we first need to characterize the \emph{pumpable paths} that uniquely 
determine the possible mean-payoff vectors in a pushdown graph.
(b)~Lemma~\ref{lem:ReductionToOneDimHardDirection} does not hold for 
arbitrary infinite-state graphs and we need to show that it does hold for 
pushdown graphs.
(c)~We require an algorithm to decide whether there is a hyperplane such that 
all the weights of the pumpable paths of a pushdown graph lie below the 
hyperplane (also for arbitrary dimensions).
The overview of our solutions to the above obstacles are as follows:
(a)~In the first part of the section (until 
Proposition~\ref{prop:GoodPathIffGoodMatrix}) we present a characterization of 
the pumpable paths in a pushdown graph.
(b)~We use Gordan's Lemma~\cite{Gordan} (a special case of Farkas' Lemma) and 
in Lemma~\ref{lemm:OneDimAndMultiDim} we prove that 
Lemma~\ref{lem:ReductionToOneDimEasyDirection} and 
Lemma~\ref{lem:ReductionToOneDimHardDirection} hold also for pushdown graphs 
(Lemma~\ref{lem:ReductionToOneDimEasyDirection} holds for 
any infinite-state graph).
(c)~Conceptually, we find the separating hyperplane by constructing a matrix 
$A$, such that every row in $A$ is a weight vector of a pumpable path, and we 
solve the linear inequality $\vec{\lambda} \cdot A <\vec{0}$.
However, in general the matrix $A$ can be of exponential size. 
Thus we need to use advanced linear-programing technique that solves in 
polynomial time linear inequalities with polynomial number of variables and 
exponential number of constraints. 
This technique requires a polynomial-time oracle that for a given $\vec{\lambda}$ returns 
a violated constraint (or says that all constraints are satisfied).
We show that in our case the required oracle is the algorithm for 
pushdown graphs with one-dimensional mean-payoff objective (which we obtain 
from~\cite{CV12}), and thus we establish a polynomial-time hyperplane 
separation technique for pushdown graphs.

\smallskip\noindent{\em Stack alphabet and commands.}
We start with the basic notion of stack alphabet and commands.
Let $\Gamma$ denote a finite set of \emph{stack alphabet}, and 
$\Com(\Gamma) = \Set{\Skip,\Pop} \cup \Set{\Push(z) \mid z\in\Gamma}$ denotes the set of 
\emph{stack commands} over $\Gamma$.
Intuitively, the command $\Skip$ does nothing, $\Pop$ deletes the top element of the stack, 
$\Push(z)$ puts $z$ on the top of the stack.
For a stack command $\com$ and a stack string $\alpha \in \Gamma^+$ 
we denote by $\com(\alpha)$ the stack string obtained by executing the 
command $\com$ on $\alpha$ (in a stack string the top denotes the right end of the string).

\smallskip\noindent{\bf Multi-weighted pushdown systems.}
A \emph{multi-weighted pushdown system (WPS)} (or a multi-weighted pushdown graph) is a tuple:
\[
\wps = \atuple{ Q, \Gamma, q_0\in Q, E \subseteq  (Q\times\Gamma) \times (Q\times \Com(\Gamma)), w:E \to \Z^k},
\]
where $Q$ is a finite set of \emph{states} with $q_0$ as the initial state; 
$\Gamma$ the finite \emph{stack alphabet} and we assume there is a special 
initial stack symbol $\bot \in \Gamma$; $E$ describes the set of 
edges or transitions of the pushdown system; and $w$ is a weight function 
that assigns an integer weight vector to every edge;
we denote by $w_i$ the projection of $w$ to the $i$-th dimension.
We assume that $\bot$ can be neither put on nor removed from the stack.
A \emph{configuration} of a WPS is a pair $(\alpha,q)$ where $\alpha\in \Gamma ^+$ 
is a stack string and $q\in Q$.
For a stack string $\alpha$ we denote by $\Top(\alpha)$ the top symbol of the stack.
The initial configuration of the WPS is $(\bot,q_0)$.
We use $W$ to denote the maximal absolute weight of the edge
weights.

\smallskip\noindent{\em Successor configurations and runs.}
Given a WPS $\wps$, a configuration $c_{i+1}=(\alpha_{i+1},q_{i+1})$
is a \emph{successor} configuration of a configuration $c_{i}=(\alpha_{i},q_{i})$,
if there is an edge $(q_i,\gamma_i,q_{i+1},\com) \in E$ such that 
$\com(\alpha_i)=\alpha_{i+1}$, where $\gamma_i=\Top(\alpha_i)$.
A \emph{path} $\pi$ is a sequence of configurations.
A path $\pi = \atuple{c_1, \dots, c_{n+1}}$ is a \emph{valid path} if for all 
$1 \leq i \leq n$ the configuration $c_{i+1}$ is a successor configuration of $c_i$
(and the notation is similar for infinite paths).
In the sequel we shall refer only to valid paths.
Let $\pi = \atuple{c_1,c_2,\dots,c_i,c_{i+1},\dots}$ be a path.
We denote by $\pi[j] = c_j$ the $j$-th configuration of the path and 
by $\pi[i_1,i_2] = \atuple{c_{i_1}, c_{i_1 + 1}, \dots,c_{i_2}}$ the 
segment of the path from the $i_1$-th to the $i_2$-th configuration.
A path can equivalently be defined as a sequence $\atuple{c_1 e_1 e_2 \dots e_n}$,
where $c_1$ is the initial configuration and $e_i$ are valid transitions.
Our goal is to obtain an algorithm that given a WPS $\wps$ decides
if there exists an infinite path $\pi$ in $\wps$ from $q_0$ such that 
$\LimAvg(\pi) \geq \VEC{0}$.


\smallskip\noindent{\em Notations.}
We shall use (i)~$\gamma$ or $\gamma_i$ for an element of $\Gamma$;
(ii)~$e$ or $e_i$ for a transition (equivalently an edge) from $E$;
(iii)~$\alpha$ or $\alpha_i$ for a string from $\Gamma^*$.
For a path $\pi = \atuple{c_1,  c_2, \dots} = \atuple{c_1 e_1 e_2 \dots}$ we denote by
(i)~$q_i$: the state of configuration $c_i$, and 
(ii)~$\alpha_i$: the stack string of configuration $c_i$.

\smallskip\noindent{\em Stack height and additional stack height of paths.}
For a path $\pi = \atuple {(\alpha_1,q_1),\dots,(\alpha_n,q_n)}$, the \emph{stack height}
of $\pi$ is the maximal height of the stack in the path, i.e., 
$\SH(\pi) = \max\Set{|\alpha_1|,\dots,|\alpha_n|}$.
The \emph{additional stack height} of $\pi$ is the additional height 
of the stack in the segment of the path, i.e., the additional 
stack height $\ASH(\pi)$ is 
$\SH(\pi) - \max\Set{|\alpha_1|,|\alpha_n|}$.

\smallskip\noindent{\em Pumpable pair of paths.}
Let $\pi = \atuple{c_1 e_1 e_2 \dots}$ be a finite or infinite path.
A \emph{pumpable pair of paths} for $\pi$ is a pair of non-empty sequences of edges:
$(p_1,p_2)=(e_{i_1}e_{i_1+1}\dots e_{i_1 + n_1},  e_{i_2}e_{i_2+1}\dots e_{i_2 + n_2})$, 
for $n_1,n_2\geq 0$, $i_1 \geq 0$ and $i_2 > i_1+n_1$ such that for every $j\geq 0$ the path
$\pi_{(p_1,p_2)}^j$ obtained by pumping the pair of paths $p_1$ and $p_2$ for $j$ times each is a valid path, i.e., 
for every $j \geq 0$ we have
\[
\pi_{(p_1,p_2)}^j = \atuple{
c_1 e_1 e_2 \dots e_{i_1 - 1} (e_{i_1} e_{i_1 + 1}\dots e_{i_1 + i_n})^j e_{i_1 + i_n + 1} \dots e_{i_2 - 1} (e_{i_2} e_{i_2 + 1} \dots e_{i_2 + n_2})^j e_{i_2 + n_2} \dots}
\]
is a valid path.
We will show that large additional stack height implies the 
existence of pumpable pair of paths. To prove the results we need the notion
of \emph{local minimum} of paths.

\smallskip\noindent{\em Local minimum of a path.}
Let $\pi = \atuple{c_1, c_2, \dots}$ be a path.
A configuration $c_i = (\alpha_i,q_i)$ is a \emph{local minimum} if for every 
$j\geq i$ we have $\alpha_i \sqsubseteq \alpha_j$ (i.e., the stack string 
$\alpha_i$ is a prefix string of $\alpha_j$).
One basic fact is the every infinite path has infinitely many local minimum.
We discuss the proof of the basic fact and some properties of local minimum.
Consider a path $\pi=\atuple{c_1,c_2,\dots}$. 
If there is a finite integer $j$ such that from some point on 
(say after $i$-th index) the stack height is always at least $j$, 
and the stack height is $j$ infinitely often,
then every configuration after $i$-th index with stack height $j$ is a 
local minimum (and there are infinitely many of them).
Otherwise, for every integer $j$, there exists an index $i$, such that
for every index after $i$ the stack height exceeds $j$, and then 
for every $j$, the last configuration with stack height $j$ is a 
local minimum and we have infinitely many local minimum. 
This shows the basic fact about infinitely many local minimum of a 
path.
We now discuss a property of consecutive local minimum in a path.
If we consider a path and the sequence of local minimum, and let
$c_i$ and $c_j$ be two consecutive local minimum. 
Then either $c_i$ and $c_j$ have the same stack height, or else
$c_j$ is obtained from $c_i$ with one push operation.

\smallskip\noindent{\em Non-decreasing paths and cycles, and proper cycles.}
A path from configuration $(\alpha\gamma,q_1)$ to configuration 
$(\alpha\gamma\alpha_2,q_2)$ is a \emph{non-decreasing $\alpha$-path} 
if $(\alpha\gamma,q_1)$ is a local minimum.
Note that if $\pi$ is a non-decreasing $\alpha$-path for some 
$\alpha\in\Gamma^*$, then the same sequence of transitions leads to a 
non-decreasing $\beta$-path for every $\beta\in\Gamma^*$.
Hence we say that $\pi$ is a non-decreasing path if there exists 
$\alpha \in \Gamma^*$ such that $\pi$ is a non-decreasing $\alpha$-path.
A \emph{non-decreasing cycle} is a non-decreasing path from $(\alpha_1, q)$ to $(\alpha_2, q)$ 
such that the top symbols of $\alpha_1$ and $\alpha_2$ are the same.
A non-decreasing cycle from  $(\alpha_1, q)$ to $(\alpha_2, q)$  
is a \emph{proper cycle} if $\alpha_1=\alpha_2$ (i.e., returns to the same configuration).
By convention, when we say that a path $\pi$ is a non-decreasing path from $(\gamma_1,q_1)$ to $(\gamma_2,q_2)$, 
it means that for some $\alpha_1,\alpha_2\in \Gamma^*$, the path $\pi$ is a non-decreasing path from 
$(\alpha_1\gamma_1,q_1)$ to $(\alpha_1\gamma_1\alpha_2\gamma_2,q_2)$.

\smallskip\noindent{\em Cone of pumpable pairs.}
We denote $\Rplus = [0,+\infty)$.
For a finite non-decreasing path $\pi$ we denote by $\PumpSet(\pi)$ the (finite) set of 
pumpable pairs that occur in $\pi$, that is, $\PumpSet(\pi) = \set{(p_1,p_2)\in (E^*\times E^*)\mid \mbox{$p_1$ and $p_2$ are a pumpable pair in $\pi$}}$.
Let $\PumpSet(\pi)=\set{P_1=(p_1^1,p_2^1), P_2=(p_1^2,p_2^2), \ldots, P_j=(p_1^j,p_2^j)}$,
and we denote by $\PumpMatrix(\pi)$ the matrix that is formed by the weight vectors of the pumpable pairs of $\pi$, that is, 
the matrix has $j$ rows and the $i$-th row of the matrix is $w(p_1^i) + w(p_2^i)$ (every weight vector is a row in the matrix).
We denote by $\PumpCone(\pi)$ the cone of the weight vectors in $\PumpSet(\pi)$, 
formally, $\PumpCone(\pi) = \set{\PumpMatrix(\pi) \cdot \VEC{x}\mid \VEC{x}\in(\Rkplus \backslash \{\VEC{0}\})}$.

\begin{figure}
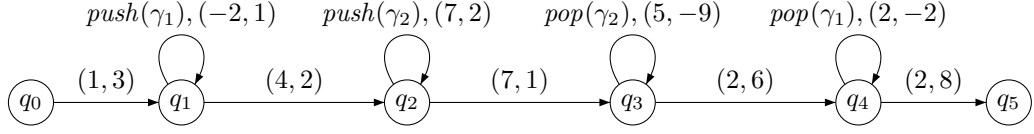

	 \begin{center} 
    \begin{gpicture}(25, 10)(0,-5)
    \thinlines

    \node[Nw=6.0,Nh=6.0](A1)(-50,0){$q_0$}

    \node[Nw=6.0,Nh=6.0](A2)(-30,0){$q_1$}

    \node[Nw=6.0,Nh=6.0](A3)(0,0){$q_2$}

    \node[Nw=6.0,Nh=6.0](A4)(30,0){$q_3$}

    \node[Nw=6.0,Nh=6.0](A5)(60,0){$q_4$}

    \node[Nw=6.0,Nh=6.0](A6)(80,0){$q_5$}

    \drawloop[loopangle=90,loopdiam=6](A2){$\mathit{push}(\gamma_1),(-2,1)$}
    \drawloop[loopangle=90,loopdiam=6](A3){$\mathit{push}(\gamma_2),(7,2)$}

    \drawloop[loopangle=90,loopdiam=6](A4){$\mathit{pop}(\gamma_2),(5,-9)$}
    \drawloop[loopangle=90,loopdiam=6](A5){$\mathit{pop}(\gamma_1),(2,-2)$}



    \drawedge(A1,A2){$(1,3)$}
    \drawedge(A2,A3){$(4,2)$}
    \drawedge(A3,A4){$(7,1)$}
    \drawedge(A4,A5){$(2,6)$}
    \drawedge(A5,A6){$(2,8)$}

    \end{gpicture}
\end{center}
  \caption{A WPS $\wps$. If an edge is not label with a command, then the command is $\mathit{skip}$. The label $\mathit{pop}(\gamma)$ stands for: if the top symbol is $\gamma$, then a pop transition is possible.
}
\label{fig:wpsG1}
\end{figure}

\smallskip\noindent{\bf Example.}
We illustrate the definitions with the aid of an example.
Consider the WPS shown in Figure~\ref{fig:wpsG1}.
Consider all the possible paths from $(\bot,q_0)$ to $(\bot,q_5)$.
Every such path is of the form
\[q_0\to q_1 \to (q_1\to q_1)^m \to q_2 \to (q_2 \to q_2)^n \to q_3 \to (q_3 \to q_3)^n \to q_4 \to (q_4 \to q_4)^m \to q_5\]
for some non-negative numbers $m$ and $n$.
Hence there are two pumpable pairs, namely, $P_1=(q_1\to q_1,q_4\to q_4)$ and $P_2=(q_2\to q_2,q_3\to q_3)$.
Given the weight function $w$ (as shown in the figure) we have $w(P_1)=(0,-1)$ and $w(P_2)=(12,-7)$.
Therefore we have the following:
\begin{itemize}
\item $\PumpSet((\bot,q_0),(\bot,q_5)) = \{(q_1\to q_1,q_4\to q_4),(q_2\to q_2,q_3\to q_3)\}$; 
\item $\PumpMatrix((\bot,q_0),(\bot,q_5)) = 
\begin{pmatrix}
0 & -1 \\
12 & -7\end{pmatrix}$; and 
\item $\PumpCone((\bot,q_0),(\bot,q_5)) = \{x_1\cdot (0,-1) + x_2\cdot (12,-7)\mid x_1,x_2\geq 0 \wedge x_1+ x_2 > 0\}$ (see Figure~\ref{fig:Pcone}).
\end{itemize}
The example illustrates the various concepts we have introduced.

\begin{figure}
\setlength\unitlength{1cm}
\begin{center}
\begin{pspicture}(1,1)
{\color{gray}\polygon*(0.5,0.5)(0.5,-0.5)(1.5,-0.0833)}
\psline[linewidth=0.5pt](0.5,0.5)(1.25,0.5)
\psline[linewidth=0.5pt](0.5,0.5)(-0.25,0.5)
\psline[linewidth=0.5pt](0.5,0.5)(0.5,-0.25)
\psline[linewidth=0.5pt](0.5,0.5)(0.5,1.25)

\psline{->}(0.5,0.5)(0.5,-0.4)
\psline{->}(0.5,0.5)(1.4,-0.025)
\psline[linestyle=dotted,dotsep=1pt](0.5,-0.4)(0.5,-0.6)
\psline[linestyle=dotted,dotsep=1pt](1.4,-0.025)(1.6,-0.1416)

\psline[linestyle=dotted,dotsep=1pt](1,-0.29165)(1.075,-0.410375)
\end{pspicture}
\end{center}
  \caption{$\PumpCone((\bot,q_0),(\bot,q_5))$}
\label{fig:Pcone}
\end{figure}

\smallskip\noindent{\em Notations and abbreviations.}
Fix $\ell = (|Q|\cdot |\Gamma|)^{(|Q|\cdot |\Gamma|)^2+1}$ for the rest of the section.
For $q_1,q_2\in Q$ and $\gamma_1,\gamma_2\in\Gamma$, by abuse of notation we denote by 
$\PumpSet((\gamma_1,q_1),(\gamma_2,q_2))$ the (finite) set of all pumpable pair of paths, 
not longer than $\ell$, that occur in a non-decreasing path from $(\gamma_1,q_1)$ to $(\gamma_2,q_2)$;
we similarly define $\PumpMatrix((\gamma_1,q_1),(\gamma_2,q_2))$ and 
$\PumpCone((\gamma_1,q_1),(\gamma_2,q_2))$.
If $q_1 = q_2$ and $\gamma_1 = \gamma_2$, then we 
abbreviate $\PumpSet((\gamma_1,q_1),(\gamma_1,q_1))$ by 
$\PumpSet((\gamma_1,q_1))$, and similarly for $\PumpMatrix$ and $\PumpCone$.
The next lemma was proved in~\cite{CV12}.

\begin{lem}[\cite{CV12}]\label{lemm:EveryDeepPathHasAPumpablePair}
Let $\pi$ be a finite path such that $\ASH(\pi) > (|Q|\cdot|\Gamma|)^2$.
Then $\pi$ has a pumpable pair of paths.
\end{lem}

In the next lemma we show that any sufficiently long non-decreasing path contains a pumpable pair of paths.

\begin{lem}\label{lemm:LongPathHasPumpablePair}
Every non-decreasing path longer than $\ell$ has a pumpable pair of paths.
\end{lem}
\begin{proof}
Let $\pi$ be a non-decreasing path longer than $\ell$.
If $\ASH(\pi) > (|Q|\cdot|\Gamma|)^2$, then by Lemma~\ref{lemm:EveryDeepPathHasAPumpablePair} we get the desired result;
otherwise, it is an easy observation that $\pi$ contains a proper cycle, which is by definition a pumpable pair of paths (where one path in the pair is empty).
\end{proof}

\begin{cor}\label{cor:LongPathHasShortPumpablePair}
Every non-decreasing path longer than $\ell$ has a pumpable pair of paths with length at most $\ell$.
\end{cor}

The next two lemmas show basic properties of $\PumpSet$.
The first lemma asserts that we can decompose every non-decreasing path to a set of pumpable pairs and a short non-decreasing path.

\begin{lem}\label{lemm:Decomposable}
For every non-decreasing path $\pi$ from $(\gamma_1,q_1)$ to $(\gamma_2,q_2)$ 
there exists a tuple of pumpable pair of paths $P_1 = (p^1_1,p^1_2), P_2 = (p^2_1,p^2_2),\dots, P_j = (p^j_1,p^j_2)\in 
\PumpSet((\gamma_1,q_1),(\gamma_2,q_2))$ each of length at most $\ell$ (i.e., for all $1\leq i \leq j$ we have 
$|P_i|\leq \ell$), a finite non-decreasing path $\pi_0$ from $(\gamma_1,q_1)$ 
to $(\gamma_2,q_2)$ with length at most $\ell$, and non-negative constants $m_1,\dots,m_j$ such that 
$w(\pi) = w(\pi_0) + \sum_{i=1}^j m_i \cdot w(P_i)$ and $|\pi| = |\pi_0| + \sum_{i=1}^j m_i \cdot |P_i|$.
\end{lem}
\begin{proof}
The proof is by 
induction of the length of $\pi$.
If $|\pi| \leq \ell$, then we are done by choosing $j=0$ and $\pi_0=\pi$.
Otherwise, by Corollary~\ref{cor:LongPathHasShortPumpablePair}, 
the path has a pumpable pair $P=(p_1,p_2)$ with length less than $\ell$ 
(and hence $P\in \PumpSet((q_1,\gamma_1),(q_2,\gamma_2))$).
Let $\pi^*$ be the path that is obtained from $\pi$ by pumping $P$ zero times
(i.e., $\pi^*$ is obtained by omitting $P$ from $\pi$);
clearly $\pi^*$ is a non-decreasing path from $(\gamma_1,q_1)$ to 
$(\gamma_2,q_2)$ and shorter than $\pi$, any by the induction hypothesis we 
get the desired result. 
\end{proof}

The following lemma shows the connection between the average weight of a path and $\PumpSet$.

\begin{lem}\label{lemm:MainProp}
If $\PumpCone((\gamma_1,q_1),(\gamma_2,q_2)) \cap \Rkplus = \emptyset$, then there exist constants 
$\epsilon > 0$ and $m\in\Nat$, such that for every finite non-decreasing path $\pi$ from $(\gamma_1,q_1)$ 
to $(\gamma_2,q_2)$, there exists a dimension $t$ such that $w_t(\pi) \leq m - \epsilon\cdot |\pi|$.
\end{lem}
\begin{proof}
In order to define $\epsilon$, we consider the following linear programming problem with the variables 
$x_1,x_2,\dots$ and $r$: the objective function is to maximize $r$ subject to the constraints below
\begin{equation}
\sum_{z \in\PumpSet((\gamma_1,q_1),(\gamma_2,q_2))} x_z \cdot w_t(z) \geq r \qquad \mbox{ for $t=1,\dots,k$}
\end{equation}
\begin{equation}
\sum_{z\in\PumpSet((\gamma_1,q_1),(\gamma_2,q_2))} x_z = 1
\end{equation}
\begin{equation}
x_z \geq 0 \qquad \mbox{ for all $z\in\PumpSet((\gamma_1,q_1),(\gamma_2,q_2))$}
\end{equation}
Intuitively, the first constraint specifies that there is a convex combination 
of the weights of the pumpable pairs to ensure at least $r$ in every dimension;
and the following two constraints is to ensure that it is a convex combination.
As the domain of the variables is closed and bounded, there exists a maximum value to the linear program, and
let $r^*$ be the maximum value.
If $r^* \geq 0$, then we get a contradiction to the assumption that $\PumpCone((\gamma_1,q_1),(\gamma_2,q_2)) \cap \Rkplus = \emptyset$.
Hence we have $r^* < 0$.
We define $m = (\ell+1) \cdot W - r^*$, $\epsilon = -\frac{r^*}{\ell}$ and 
we claim that for every non-decreasing path 
$\pi$ from $(\gamma_1,q_1)$ to $(\gamma_2,q_2)$ there is a dimension $t$ such 
that $w_t(\pi)\leq m - \epsilon\cdot |\pi|$.

By Lemma~\ref{lemm:Decomposable}, there exists a path $\pi_0$ with 
length at most $\ell$, a (finite) sequence of pumpable pairs 
$P_1,\dots,P_j \in \PumpSet((\gamma_1,q_1),(\gamma_2,q_2))$ each of length at most 
$\ell$ and constants $m_1,\dots,m_j$ such that 
$w(\pi) = w(\pi_0) + \sum_{i=1}^j m_i \cdot w(P_i)$ 
and $|\pi| = |\pi_0| + \sum_{i=1}^j m_i \cdot |P_i|$.
We define $M = \sum_{i=1}^j m_i$.
As all $|P_i|$ and $|\pi_0|$ are bounded by $\ell$, we get that 
$M\geq \frac{|\pi| - \ell}{\ell}$.
Observe that if we set $x_i = \frac{m_i}{M}$ for $i=1$ to $j$, and let 
$x_z=0$ for all other $z \in \PumpSet((\gamma_1,q_1),(\gamma_2,q_2))$,
then they satisfy the constraints for convex combination.
Hence there must exist a dimension $t$ for which 
$\frac{1}{M} \sum_{i=1}^j m_i \cdot w_t(P_i) \leq r^*$ (since $r^*$ 
is the maximum among the feasible solutions).
Thus $w_t(\pi) \leq w_t(\pi_0) + M \cdot r^*$ and since $r^* < 0$ we have 
\[
w_t(\pi) \leq w_t(\pi_0) + \frac{|\pi|-\ell}{\ell} \cdot r^* =
w_t(\pi_0) -r^* + |\pi| \cdot \frac{r^*}{\ell}.
\]
Therefore, for the choice of $m \geq \ell \cdot W -r^*$ and $\epsilon = -\frac{r^*}{\ell}$,
we obtain the desired result.
\end{proof}

The next proposition gives a sufficient and necessary condition for the existence of a path with non-negative mean-payoff 
values in all the dimensions.

\begin{prop}\label{prop:GoodPathIffGoodMatrix}
There exists an infinite path $\pi$ such that $\LimInfAvg(\pi) \geq \VEC{0}$ if and only if there exists a (reachable) non-decreasing cycle $\pi$ such that $\Rkplus\cap\PumpCone(\pi) \neq \emptyset$.
\end{prop}
\begin{proof}
We first prove the direction from right to left.
If there exists a path $\pi$ such that $\Rkplus\cap\PumpCone(\pi)\neq\emptyset$,
then by definition there are $j$ pumpable pairs $P_1, P_2,\ldots,P_j$ 
with weight vectors $y_1=w(P_1),\dots,y_j=w(P_j)$ such that there exist $j$ 
positive constants (w.l.o.g natural positive constants) $n_1,\dots,n_j$ such that $\sum_{i=1}^j n_i \cdot y_i \geq \VEC{0}$.
For every $a,b\in\Nat$ we denote by $\pi^{a,b}$ the (finite) path that is formed by pumping the $a$-th pumpable pair $b$ times.
We denote by $\wh{\pi}^b =\pi^{1,b\cdot n_1} \cdot \pi^{2,b\cdot n_2} \cdot \cdots \cdot \pi^{j,b\cdot n_j} \cdots$, where
the $i$-th pumpable pair is pumped $b\cdot n_i$ times, respectively.
We note that $\pi^{a,b}$ is a non-decreasing cycle, and for the infinite path
\[
\pat^* = \wh{\pi}^1 \cdot \wh{\pi}^2 \cdot \wh{\pi}^3 \cdots \wh{\pi}^b \cdots
\]
we get $\LimAvg(\pat^*) \geq \VEC{0}$.
The reason we have $\LimAvg(\pat^*) \geq \VEC{0}$ is as $b$ tends to infinity, 
the average weight is determined only by the weights of the $j$ pumpable pairs 
and their coefficients $n_1,\dots,n_j$, and we have $\sum_{i=1}^j n_i \cdot y_i \geq \VEC{0}$.
This completes the proof for the direction from right to left.

For the converse direction, let $\pi$ be an infinite path such that $\LimAvg(\pi) \geq \VEC{0}$,
and let $(\gamma,q)$ be a top configuration that occurs infinitely often in the local minimum of $\pi$.
Since $\LimAvg(\pi) \geq \VEC{0}$ it follows that for every $\epsilon > 0$ there exists a 
non-decreasing cycle that begins at $(\gamma,q)$ with average weight at least $-\epsilon$ in every dimension.
Hence, by Lemma~\ref{lemm:MainProp} it follows that 
$\PumpCone(\gamma,q) \cap \Rkplus \neq \emptyset$, and hence, there exists a non-decreasing 
cycle $\pi$ that starts in $(\gamma,q)$ for which 
$\PumpCone(\pi) \cap \Rkplus \neq \emptyset$. 
\end{proof}

By Proposition~\ref{prop:GoodPathIffGoodMatrix}, we can decide whether there is 
an infinite path $\pi$ for which $\LimAvg(\pi)\geq \VEC{0}$ by checking if 
there exist a tuple $(\gamma,q) \in \Gamma \times Q$ for which there is a 
non-negative (and non-trivial) solution for the equation 
$\PumpMatrix((\gamma,q)) \cdot \VEC{x} \geq 0$.
As in Lemma~\ref{lem:ReductionToOneDimEasyDirection} by adding $k$ self-loop transitions 
with weights, where the weight of transition $i$ is $-1$ in the $i$-th 
dimension and $0$ in the other dimensions, we reduce the problem to 
finding $q$ and $\gamma$ such that there is a non-negative solution for 
$\PumpMatrix((\gamma,q)) \cdot \VEC{x} = 0$.
Inspired by the techniques of~\cite{CM93}, we present an algorithm 
that solves the problem by a reduction to a corresponding one-dimensional 
problem.
As before given a $k$-dimensional weight function $w$ and a $k$-dimensional 
vector $\vec{\lambda}$ we denote by $w\cdot \vec{\lambda}$ the one-dimensional weight function obtained by multiplying the weight vectors by 
$\vec{\lambda}$.
The reduction to one-dimensional objective requires the use of Gordan's 
lemma.

\begin{lem}[Gordan's Lemma~\cite{Gordan} (see also Lemma~2 in~\cite{Papa81})]
For a matrix $A$, either $A \cdot \vec{x} = \vec{0}$ has a non-trivial non-negative solution for 
$\vec{x}$, or there exists a vector $\vec{y}$ such that $\vec{y}\cdot A^{T}$ is negative in every 
dimension.
\end{lem}


The next lemma suggests that we can reduce the multidimensional problem to a 
corresponding one-dimensional problem.

\begin{lem}\label{lemm:OneDimAndMultiDim}
Given a WPS $\wps$ with a $k$-dimensional weight function $w$, and 
$(\gamma,q)\in \Gamma \times Q$, there exists a non-trivial non-negative solution for 
$\PumpMatrix((\gamma,q)) \cdot \VEC{x} = \VEC{0}$ if and only if for every 
$\vec{\lambda} \in \R^k$ there is a non-decreasing path from $(\gamma,q)$ 
to $(\gamma,q)$ that contains a pumpable pair $P=(p_1,p_2)$ such that 
$(w \cdot \vec{\lambda})(P)\geq 0$ (i.e., the weight of the path for 
one-dimensional weight function $w\cdot \vec{\lambda}$ is non-negative).
\end{lem}
\begin{proof}
The proof is straightforward application of Gordan's Lemma to the matrix 
$\PumpMatrix((\gamma,q))$.
\end{proof}

\begin{prop}\label{prop:WeCanFindOneDimensionalPump}
There is a polynomial time algorithm that given WPS $\wps$ with $k$-dimensional
weight function $w$, $(\gamma,q)\in \Gamma \times Q$, a vector 
$\vec{\lambda} \in \Q^k$, and a rational number $r\in\Q$ decides 
if there exists a pumpable pair of paths 
$P$ in a non-decreasing cyclic path that begins at $(\gamma,q)$ in $\wps$, 
with $\frac{(w \cdot\vec{\lambda})(P)}{|P|} > r$ and $|P|\leq \ell$, and if such 
pair exists, it returns $\frac{w(P)}{|P|}$.
\end{prop}
Intuitively, the algorithm for Proposition~\ref{prop:WeCanFindOneDimensionalPump} is based on the algorithm for solving WPSs with 
one-dimensional mean-payoff objectives. 
We postpone the technically detailed proof to Section~\ref{sec:OracleForOneDim}. 
We first show how to use the result of the proposition and a result
from linear programming to solve the problem.
We first state the result for linear programming.

\smallskip\noindent{\bf Linear program with exponential constraints and polynomial-time separating oracle.}
Consider a linear program over $n$ variables and exponentially many constraints in $n$.
Given a polynomial time \emph{separating oracle} that for every point in space returns in 
polynomial time whether the point is feasible, and if infeasible returns a violated
constraint, the linear program can be solved in polynomial time using the ellipsoid method~\cite{GLS81}.
We use the result to show the following result.

\begin{prop}\label{prop:PolyTimeForAxEqual0}
There exists a polynomial time algorithm that decides whether for a given state 
$q$ and a stack alphabet symbol $\gamma$ there exists a non-trivial non-negative solution for 
$\PumpMatrix((\gamma,q)) \cdot \VEC{x} = \VEC{0}$.
\end{prop}
\begin{proof}
Conceptually, given $q$ and $\gamma$, we compute a matrix $A$, such that each row in $A$ 
corresponds to the average weight vector of a row in $\PumpMatrix((\gamma,q))$ 
(that is, the weight of a pumpable pair divided by its length), and solves the following 
linear programming problem:
For variables $r$ and $\vec{\lambda} = (\lambda_1,\dots,\lambda_k)$, the objective 
function is to minimize $r$ subject to the constraints below:
\begin{equation}
\vec{\lambda}\cdot A^T\leq \VEC{r} \quad \mbox{ where $\VEC{r} = (r,r,\dots,r)^T$}
\end{equation}
\begin{equation}
\sum_{i=1}^k \lambda_i = 1
\end{equation}
Once the minimal $r$ is computed, by Lemma~\ref{lemm:OneDimAndMultiDim}, 
there exists a solution for $\PumpMatrix((\gamma,q)) \cdot \VEC{x} = 0$ 
if and only if $r \geq 0$.

The number of rows of $A$ in the worst case is exponential (to be precise at most $\ell\cdot (2\cdot W \cdot \ell)^k$,
since the length of the path is at most $\ell$, the sum of weights is between $-W\cdot \ell$ and 
$W\cdot \ell$ and there are $k$ dimensions).
However, we do not enumerate the constraints of the linear programming problem explicitly but 
use the result of linear programs with polynomial time separating oracle.
By Proposition~\ref{prop:WeCanFindOneDimensionalPump} we have an algorithm that verifies 
the feasibility of a solution (that is, an assignment for $\vec{\lambda}$ and $r$) 
and if the solution is infeasible it returns a constraint that is not satisfied by the solution.
Thus the result of Proposition~\ref{prop:WeCanFindOneDimensionalPump} provides the 
desired polynomial-time separating oracle and we have the desired result.
\end{proof}

Hence, we get the following theorem.
\begin{thm}\label{thm:LimInfInP}
Given a WPS $\wps$ with multidimensional weight function $w$, we can decide in 
polynomial time whether there exists a path $\pi$ such that $\LimAvg(\pi) \geq \VEC{0}$.
\end{thm}

\subsection{Technical detailed proof of Proposition~\ref{prop:WeCanFindOneDimensionalPump}}\label{sec:OracleForOneDim}
In this section we prove Proposition~\ref{prop:WeCanFindOneDimensionalPump}.
Throughout this section, we assume WLOG that $\vec{\lambda}$ is a vector of integers and that $r = 0$.
Intuitively the solution is very similar to solving WPS with one-dimensional objectives,
with some technical and tedious modifications.
We will present the relevant details.
Let $\wps$ be a WPS with $k$-dimensional weight function $w$, and $w\cdot \vec{\lambda}$ be the 
one-dimensional weight function.
Let $d = (|Q|\cdot |\Gamma|)^2 + 1$.
We now recall the notion of summary function as defined in~\cite{CV12}.
In the definition of summary function below we consider the weight function 
$w\cdot \vec{\lambda}$.

\smallskip\noindent{\em Summary function.}
Let $\wps$ be a WPS. 
For $\alpha\in \Gamma^*$ we define $s_{\alpha} : Q \times \Gamma \times Q \to \Set{-\infty} \cup \Z \cup \Set{\omega}$ as following.
\begin{enumerate}
\item $s_\alpha(q_1,\gamma,q_2) = \omega$ iff for every $n\in\Nat$ there 
exists a non-decreasing path from $(\alpha \gamma, q_1)$ to $(\alpha \gamma, q_2)$ 
with weight at least $n$.
\item $s_\alpha(q_1,\gamma,q_2) = z\in\Z$ iff the weight of the maximum 
weight non-decreasing path from configuration $(\alpha \gamma, q_1)$ to 
configuration $(\alpha \gamma, q_2)$ is $z$.
\item $s_\alpha(q_1,\gamma,q_2) = -\infty$ iff there is no non-decreasing path 
from $(\alpha \gamma, q_1)$ to $(\alpha \gamma, q_2)$. 
\end{enumerate}

\begin{rem}\label{rem:IndependentOfAlpha}
For every $\alpha_1, \alpha_2\in \Gamma^*$: $s_{\alpha_1} \equiv s_{\alpha_2}$.
\end{rem}
Due to Remark~\ref{rem:IndependentOfAlpha} it is enough to consider only $s \equiv s_{\bot}$.
The computation of the summary function will be achieved by considering 
stack height bounded summary functions defined below.

\smallskip\noindent{\em Stack height bounded summary function.}
For every $d\in\Nat$, the \emph{stack height bounded summary function} 
$s_d : Q \times \Gamma \times Q \to \Set{-\infty} \cup \Z \cup \Set{\omega}$ 
is defined as follows: 
(i)~$s_d(q_1,\gamma,q_2) = \omega$ iff for every $n\in\Nat$ there exists a 
non-decreasing path from $(\bot\gamma, q_1)$ to $(\bot \gamma, q_2)$ with 
weight at least $n$ and additional stack height at most $d$;
(ii)~$s_d(q_1,\gamma,q_2) = z$ iff the weight of the maximum weight 
non-decreasing path from $(\bot\gamma,q_1)$ to $(\bot\gamma,q_2)$ with 
additional stack height at most $d$ is $z$; and
(iii)~$s_d(q_1,\gamma,q_2) = -\infty$ iff there is no non-decreasing path 
with additional stack height at most $d$
from $(\bot \gamma, q_1)$ to $(\bot \gamma, q_2)$.
Before presenting the key lemma we recall the computation of $s_{i+1}$ from 
$s_i$ that will also introduce the relevant notions required for the lemma.

\smallskip\noindent\emph{Computation of $s_{i+1}$ from $s_{i}$ and $\wps$.}
Let $G_{\wps}$ be the finite weighted graph that is formed by all the 
configurations of $\wps$ with stack height either one or two, that is, 
the vertices are of the form $(\alpha, q)$ where $q\in Q$ and $\alpha \in 
\Set{\bot \cdot \gamma, \bot \cdot \gamma_1 \cdot \gamma_2 
\mid \gamma, \gamma_1, \gamma_2 \in \Gamma}$.
The edges (and their weights) are according to the transitions of $\wps$: 
formally, 
(i)~(Skip edges): for vertices $(\bot \cdot \alpha, q)$ we have an edge 
to $(\bot \cdot \alpha,q')$ iff $e=(q, \Top(\alpha), \Skip,q')$ is an edge
in $\wps$ (and the weight of the edge in $G_{\wps}$ is 
$(w\cdot \vec{\lambda})(e)$) where 
$\alpha=\gamma$ or $\alpha=\gamma_1\cdot \gamma_2$ for 
$\gamma,\gamma_1,\gamma_2 \in \Gamma$; 
(ii)~(Push edges): for vertices $(\bot \cdot \gamma, q)$ we have an edge to 
$(\bot \cdot \gamma \cdot \gamma', q')$ iff 
$e=(q,\gamma, \Push(\gamma'),q')$ is an edge in $\wps$ 
(and the weight of the edge in $G_{\wps}$ is $(w\cdot \vec{\lambda})(e)$)  
for $\gamma,\gamma'\in \Gamma$; and
(iii)~(Pop edges): for vertices $(\bot \cdot \gamma \cdot \gamma', q)$
we have an edge to $(\bot \cdot \gamma, q')$ iff 
$e=(q,\gamma', \Pop, q')$ is an edge in $\wps$ 
(and the weight of the edge in $G_{\wps}$ is $(w\cdot \vec{\lambda})(e)$)  
for $\gamma,\gamma'\in \Gamma$.
Intuitively, $G_{\wps}$ allows skips, push pop pairs, and only one 
additional push.
Note that $G_{\wps}$ has at most $2\cdot |Q|\cdot |\Gamma|^2$ vertices, 
and can be constructed in polynomial time.

For every $i \geq 1$, given the function $s_i$, the graph $G_{\wps}^i$ is 
constructed from $G_{\wps}$ as follows: adding edges 
$((\bot \gamma_1 \gamma_2, q_1), (\bot \gamma_1 \gamma_2, q_2))$ (if 
the edge does not exist already) and changing its weight to 
$s_i(q_1,\gamma_2,q_2)$ for every $\gamma_1, \gamma_2 \in \Gamma$ and 
$q_1, q_2 \in Q$.
The value of $s_{i+1}(q_1,\gamma,q_2)$ is exactly the weight of the 
maximum weight path between $(\bot \gamma, q_1)$ and $(\bot \gamma, q_2)$ in 
$G_{\wps}^i$ (with the following convention: $-\infty < z < \omega$, 
$z + \omega = \omega$ and $z + -\infty = \omega + -\infty = -\infty$ 
for every $z\in\Z$). 
If in $G_{\wps}^i$ there is a path from $(\bot \gamma, q_1)$ to 
$(\bot \gamma, q_2)$ that contains a cycle with positive weight,
then we set $s_{i+1}(q_1,\gamma,q_2) = \omega$.
Hence, given $s_i$ and $\wps$, the construction of $G_{\wps}^i$ is 
achieved in polynomial time, and the computation of $s_{i+1}$ 
is achieved using the Bellman-Ford algorithm~\cite{CLRS-Book} in 
polynomial time 
(the maximum weight path is the shortest weight if we define the edge length 
as the negative of the edge weight).
Also note that the Bellman-Ford algorithm reports cycles with positive weight 
(that is, negative length) which is required to set $\omega$ values of 
$s_{i+1}$.
It follows that we can compute $s_{i+1}$ given $s_i$ and $\wps$ in 
polynomial time.
In the computation of the summary function $s_i$ we also store 
along with $s_i(q_1,\gamma,q_2)$ the weight vector $w(P)$ and the 
length $|P|$ of a witness path $P$ that is maximal weight (according
to $w\cdot \vec{\lambda}$) shortest non-decreasing path from $(\gamma,q_1)$ 
to $(\gamma,q_2)$ with additional stack height at most $i$. 
We denote by $\VECT(s_i(q_1,\gamma,q_2))$ the tuple $(w(P),|P|)$.

\begin{lem}\label{lemm:FindingPumpInSameLevel}
Let $q_1,q_2\in Q$, $\gamma\in\Gamma$ and $d > (|Q|\cdot|\Gamma|)^2$, such that $s_{d}(q_1,\gamma,q_2) > 
s_{d-1}(q_1,\gamma,q_2)$, and let $\pi$ be the shortest non-decreasing path from $(\bot\gamma,q_1)$ to $(\bot\gamma,q_2)$ 
with weight $s_{d+1}(q_1,\gamma,q_2)$ and additional stack height $d$, then the following assertions hold:
\begin{enumerate}
\item The path $\pi$ contains a pumpable pair of paths $P = (p_1,p_2)$ with $(w\cdot \vec{\lambda})(P) > 0$ with 
length at most $\ell$. 
\item We can compute $w(P)$, and $\frac{w(P)}{|P|}$ in polynomial time.
\end{enumerate}
\end{lem}
\begin{proof}
The first item was proved in~\cite{CV12}.
For the second item, we consider the graphs $G^i_\wps$ as defined above.
Then for $G^d_\wps$, we compute (based on the summary function $s_d$) the 
maximum weight non-decreasing path $\rho$ 
from $(\bot\gamma,q_1)$ to $(\bot\gamma,q_2)$.
In the path $\rho$, we find a sub-path of the form 
$(\bot\gamma,z),(\bot\gamma\delta,q'),(\bot\gamma\delta,q''),(\bot\gamma,z')$, 
for which 
\begin{itemize}
\item $s_d(z,\gamma,z') > s_{d-1}(z,\gamma,z')$; and
\item $s_{d-1}(q',\delta,q'') > s_{d-2}(q',\delta,q'')$;
\end{itemize}
(note that by definition such sub-path must exist).
We store the value of the maximum weight paths from $(\bot\gamma,q_1)$ to  
$(\bot\gamma,z)$, and from $(\bot\gamma,z')$ to $(\bot\gamma,q_2)$.
We also store the $\Push$ and $\Pop$ transitions and the corresponding vector of the weight 
function $w$, and repeat the process, recursively, 
for the maximum weight non-decreasing path from $(\delta,q')$ to $(\delta,q'')$ with $\ASH(d-1)$.
We end up with a description of length $O(d)$ of the form 
\[
\begin{array}{rcl}
\rho^* & = & (\bot\gamma_1,q^1_1)\stackrel{\rho_1}{\leadsto}(\bot\gamma_1,q^1_2)
\stackrel{\Push_1}{\to} (\bot\gamma_1\gamma_2,q^2_1)\stackrel{\rho_2}{\leadsto}
(\bot\gamma_1\gamma_2,q^2_2)\stackrel{\Push_2}{\to} (\bot\gamma_1\gamma_2\gamma_3,q^3_1) \stackrel{\rho_3}{\leadsto} 
(\bot\gamma_1\gamma_2\gamma_3,q^3_2) \stackrel{\Push_3}{\to} \\[2ex] 
 & & \cdots \stackrel{\Push_d}{\to} (\bot\gamma_1\dots\gamma_d,q^d_1) \stackrel{\rho_{d+1}}{\leadsto} (\bot\gamma_1\dots\gamma_d,q^d_2) \stackrel{\Pop_1}{\to} 
 (\bot\gamma_1\dots\gamma_{d-1},q^{d-1}_3)\stackrel{\rho_{d+2}}{\leadsto}(\bot\gamma_1\dots\gamma_{d-1},q^{d-1}_4) 
\\[2ex] & & 
\stackrel{\Pop_2}{\to}
(\bot\gamma_1\dots\gamma_{d-2},q^{d-2}_3)\stackrel{\rho_{d+3}}{\leadsto}(\bot\gamma_1\dots\gamma_{d-2},q^{d-2}_4)
\stackrel{\Pop_3}{\to}
\dots \stackrel{\Pop_d}{\to} (\bot\gamma_1,q^1_3) \stackrel{\rho_{2\cdot d+1}}{\leadsto} (\bot\gamma_1,q^1_4);
\end{array}
\]
where $q^1_1 = q_1$, $q^1_4 = q_2$ and $\gamma_1 = \gamma$.
Intuitively, the path $\rho^*$ is decomposed as the path 
$\rho_1 \ \Push_1 \ \rho_2 \ \Push_2 \ \cdots \ \Push_d \ \rho_{d+1} \ \Pop_1 \ \rho_{d+2}\ \cdots \ \Pop_d\  \rho_{2\cdot d+1}$,
where the $\rho_1$ realizes the value $s_d(q_1^1,\gamma_1,q_2^1)$, $\rho_2$ realizes the value 
$s_{d-1}(q_1^2,\gamma_2,q_2^2)$ and so on; and similarly $\rho_{d+1}$ realizes the value $s_0(q_1^d,\gamma_d,q_2^d)$, $\rho_{d+2}$ realizes the value $s_1(q_3^{d-1},\gamma_{d-1},q_4^{d-1})$, $\rho_{d+3}$ realizes the value $s_2(q_3^{d-2},\gamma_{d-2},q_4^{d-2})$ and so on;
and finally, $\rho_{2\cdot d + 1}$ realizes $s_d(q_3^1,\gamma_1,q_4^1)$.

Since $d > (|Q|\cdot |\Gamma|)^2$, there must exist $1\leq i<j \leq d$, 
and $h_1,h_2,h_3,h_4 \in \RangeSet{1}{4}$ such that $q^i_{h_1} = q^j_{h_2}$, $q^i_{h_3} = q^j_{h_4}$, 
$\gamma_i = \gamma_j$, and the weight of the path from 
$(\bot\gamma_1\dots\gamma_i,q^i_{h_1})$ to $(\bot\gamma_1\dots\gamma_j,q^j_{h_2})$ 
plus the weight of the path from $(\bot\gamma_1\dots\gamma_j,q^i_{h_3})$ to 
$(\bot\gamma_1\dots\gamma_i,q^j_{h_4})$ is positive.
We sequentially iterate over all such tuples of $i,j,h_1,h_2,h_3$ and $h_4$ in polynomial time,
and a witness path $P$ can be obtained as of the form of $\rho^*$.
The computation of  $w(P)$ and $\frac{w(P)}{|P|}$ is obtained from the vector of the 
summary function, and the $\Push$ and $\Pop$ transitions along with the
vector of weights according to $w$ of such transitions, i.e., 
\[
\begin{array}{rcl}
(w(P),|P|) 
& = & \displaystyle  
\big(\sum_{i=1}^d w(\Push_i) + w(\Pop_i),2\cdot d\big ) +
\sum_{i=1}^{2\cdot d + 1} (w(\rho_i),|\rho_i|) 
\\[3ex]
& = & \displaystyle  
\big(\sum_{i=1}^d w(\Push_i) + w(\Pop_i),2\cdot d\big) +
\sum_{i=1}^{d+1} \VECT(s_{d + 1 - i}(q_1^i,\gamma_i,q_2^i)) +
\sum_{i=d+2}^{2d+1} \VECT(s_{i - d - 1}(q_3^i,\gamma_i,q_4^i)).
\end{array}
\]
Hence it follows that we can compute $w(P)$ and $\frac{w(P)}{|P|}$ in polynomial 
time and the proof follows. 
\end{proof}

Our goal now is the computation of the $\omega$ values of the summary function.
To achieve the computation of $\omega$ values we will define another summary 
function $s^*$ and a new WPS $\wps^*$ such that certain cycles in $\wps^*$ will 
characterize the $\omega$ values of the summary function. 
We now define the summary function $s^*$ and the pushdown system $\wps^*$. 
Let $d= (|Q|\cdot|\Gamma|)^2$. The new summary function $s^*$ is defined as follows:
if the values of $s_d$ and $s_{d+1}$ are the same then it is assigned the value 
of $s_d$, and otherwise the value $\omega$. Formally,
\[ 
s^*(q_1,\gamma,q_2) = 
   \left\{ \begin{array}{ll}
	s_d(q_1,\gamma,q_2) & \mbox{   if $s_d(q_1,\gamma,q_2) = s_{d+1}(q_1,\gamma,q_2)$} \\
	\omega & \mbox{   if $s_d(q_1,\gamma,q_2) < s_{d+1}(q_1,\gamma,q_2)$}.
   \end{array} \right. 
\]
The new WPS $\wps^*$ is constructed from $\wps$ by adding the following set of $\omega$-edges:
$\Set{(q_1,\gamma,q_2,\Skip) \mid s^*(q_1,\gamma,q_2) = \omega}$.

\begin{lem}[\cite{CV12}]\label{lemm:IfOmegaThenOmega}
For all $q_1, q_2 \in Q$ and $\gamma \in \Gamma$,
the following assertion holds:
the original summary function $s(q_1,\gamma,q_2) = \omega$ iff 
there exists a non-decreasing path in $\wps^{*}$ from 
$(\bot\gamma,q_1)$ to $(\bot\gamma,q_2)$ that goes through an $\omega$-edge.
\end{lem}
\begin{comment}
\begin{proof}
The direction from right to left is easy: if there is a non-decreasing
path in $\wps^*$ that goes through an $\omega$-edge, it means that 
there exists $(q_1',\gamma',q_2')$ with either $s_d(q_1',\gamma,q_2') =\omega$ 
or $s_d(q_1',\gamma',q_2')< s_{d+1}(q_1',\gamma',q_2')$.
If $s_d(q_1',\gamma,q_2') =\omega$, then clearly $s(q_1',\gamma,q_2') =\omega$.
Otherwise we have $s_d(q_1',\gamma',q_2')< s_{d+1}(q_1',\gamma',q_2')$, 
and then the proof of Lemma~\ref{lemm:AlmostS} shows that 
$s(q_1',\gamma',q_2')=\omega$.
Since there exists finite path from $(\bot\gamma,q_1)$ to $(\bot\gamma,q_2)$ 
with the $\omega$-edge it follows that $s(q_1,\gamma,q_2)=\omega$.

For the converse direction, we consider the case that
$s(q_1,\gamma,q_2) = \omega$.
If $s^*(q_1,\gamma,q_2) = \omega$, then the proof follows immediately.
Otherwise it follows that $s_d(q_1,\gamma,q_2) \in\Z$.
Hence there exists a weight $n\in\Z$ such that the non-decreasing path with the minimal additional 
stack height with weight $n$ has additional stack height $d' \geq d + 1$.
Let $\pi$ be that path.
Then there exists a non-decreasing subpath that starts at 
$(\alpha\gamma',q_1')$ and ends at $(\alpha\gamma',q_2')$ with additional stack 
height exactly $d+1$.
If $s_{d+1}(q_1',\gamma',q_2') = s_{d}(q_1',\gamma',q_2')$, 
then $\pi$ is not the path with the minimal additional stack height.
Hence, as $s_{d+1}(q_1',\gamma',q_2') > s_{d}(q_1',\gamma',q_2')$, 
by definition $s^*(q_1',\gamma',q_2') = \omega$ and the proof follows. 
\hfill\qed
\end{proof}
\end{comment}

We will now present the required polynomial-time algorithm for Proposition~\ref{prop:WeCanFindOneDimensionalPump},
and we present the algorithm for the case with $r=0$ (and this is without loss of generality).
The algorithm is similar to solution of WPS with one-dimensional objective of~\cite{CV12}.
The final ingredient is the notion of summary graph.

\smallskip\noindent{\em Summary graph and positive simple cycles.} Given a WPS 
$\wps=\atuple{Q, \Gamma, q_0\in Q, E \subseteq  (Q\times\Gamma) \times (Q\times \Com(\Gamma)), w\cdot \vec{\lambda}:E \to \Z}$
and the summary function $s$, we construct the \emph{summary graph} $\Gr(\wps)=(\ov{V},\ov{E})$ of $\wps$ with 
a weight function $\ov{w}: \ov{E} \to \Z \cup \Set{\omega}$ as follows: (i)~$\ov{V} = Q\times \Gamma$; and 
(ii)~$\ov{E} = E_{\Skip} \cup E_{\Push}$ where
$E_{\Skip} = \Set{((q_1,\gamma),(q_2,\gamma)) \mid s(q_1,\gamma,q_2) > -\infty}$, and
$E_{\Push} = \Set{((q_1,\gamma_1),(q_2,\gamma_2)) \mid (q_1,\gamma_1,q_2,\Push(\gamma_2)) \in E}$;
and (iii)~for all $e = ((q_1,\gamma),(q_2,\gamma))\in E_{\Skip}$ we have $\ov{w}(e) =  s(q_1,\gamma,q_2)$,
and for all $e\in E_{\Push}$ 
we have $\ov{w}(e)=(w\cdot \vec{\lambda})(e)$ (i.e., according to weight function of $\wps$).
A simple cycle $C$ in $\Gr(\wps)$ is a \emph{positive simple cycle} iff one of
the following conditions hold: (i)~either $C$ contains an $\omega$-edge (i.e.,
edge labeled $\omega$ by $\ov{w}$); or 
(ii)~the sum of the weights of the edges of the cycles according to 
$\ov{w}$ is positive.
The summary functions and the summary graph can be constructed in polynomial time.
The first step of the algorithm is to build the summary graph and to check if there is a path 
from $(\gamma,q)$ to $(\gamma,q)$ with a positive weight.
We consider the following cases of existence of such a positive weight path.
\begin{enumerate}

\item If there is no such path, then there does not exist pumpable pair of paths $P=(p_1,p_2)$ 
with positive weight (i.e., there exists no pumpable pair $P$ with $(w\cdot \vec{\lambda})(P)>0$).

\item We now consider the case when such a positive weight path exists. If such a path exist,
we consider the path with maximum weight that is shortest (i.e., among the ones with maximum 
weight we choose a path that is shortest). We have two distinct cases.
\begin{enumerate}

\item We first consider the case when the path do not go through an $\omega$ edge.
Then the path does not have a pumpable pair for the following reason:
if the pumpable pair is positive, then the weight is not the maximum,
and if the pumpable pair is non-negative, removing it ensures we obtain a 
maximum weight path with shorter length. 
Hence the length of the path is at most $\ell$.
Since we have stored the vector of the summary function (which stores 
the weights according to $w$ and length of the witness paths)
we compute the weight of this path according to $w$ 
(and not according to $w\cdot\vec{\lambda}$), and return the average weight of 
this path.
\item 
Otherwise, the path goes through an $\omega$ edge in the summary graph.
If there is an $\omega$ edge due to a proper cycle with positive weight, 
then we can detect this cycle in the construction of the summary graph and 
compute its average weight according to $w$ (since we have the vector of
the summary function that stores the weight according to $w$ and the length
of the witness paths).
Otherwise, by Lemma~\ref{lemm:IfOmegaThenOmega}, it follows that there is a non-decreasing path 
from $(\gamma,q)$ to $(\gamma,q)$ that has a non-decreasing sub-path from $(\delta,q_1)$ to $(\delta,q_2)$ 
and $s_{d+1}(q_1,\delta,q_2) > s_d(q_1,\delta,q_2)$.
We have already described a polynomial time algorithm for finding such $q_1,q_2$ and $\delta$. 
Once we find  $q_1,q_2$ and $\delta$, by Lemma~\ref{lemm:FindingPumpInSameLevel}, we can compute $w(P)$ and $\frac{w(P)}{|P|}$ in polynomial time.
\end{enumerate}
\end{enumerate}
The proof of Proposition~\ref{prop:WeCanFindOneDimensionalPump} follows.

%% file: modular.tex
\section{Recursive Games under Modular Strategies with Mean-payoff Objectives}
In this section we will consider recursive games (which are equivalent to 
pushdown games) with modular strategies.
Note that there is no intuitive interpretation of modular strategies for 
pushdown games and it is standard (as considered in all works in literature) 
to define and consider modular strategies in the context of recursive games.
We start with the definitions and present four results for mean-payoff 
objectives in such games:
(1)~we show undecidability for multidimensional problem, and hence focus
on the one-dimensional case;
(2)~for the one-dimensional case we show a NP-hardness result;
(3)~we present an algorithm that runs in polynomial time when relevant
parameters are fixed; and 
(4)~finally we show a reduction from finite-state parity games to show the 
hardness of fixed parameter tractability.

\smallskip\noindent{\bf Weighted recursive game graphs (WRGs).}
A \emph{recursive game graph} $\wrg$ consists of a tuple 
$\atuple{A_0,A_1,\dots,A_n}$ of \emph{game modules}, where each 
game module $A_i = (N_i,B_i,V^1_i,V_i^2,\En_i,\Ex_i,\delta_i)$ consists of the following components:
\begin{itemize}
\item A finite nonempty set of \emph{nodes} $N_i$.
\item A nonempty set of \emph{entry} nodes $\En_i \subseteq N_i$ and a 
nonempty set of \emph{exit} nodes $\Ex_i \subseteq N_i$.
\item A set of \emph{boxes} $B_i$.
\item Two disjoint sets $V_i^1$ and $V_i^2$ that partition the set of 
nodes and boxes into two sets, i.e., $V_i^1 \cup V_i^2 = N_i \cup B_i$ 
and $V_i^1 \cap V_i^2 = \emptyset$.
The set $V_i^1$ (resp. $V_i^2$) denotes the places where it is the 
turn of player~1 (resp. player~2) to play (i.e., choose transitions).
We denote the union of $V_i^1$ and $V_i^2$ by $V_i$.
\item A labeling $Y_i : B_i \to \RangeSet{1}{n}$ that assigns to 
every box an index of the game modules $A_1 \dots A_n$.
\item Let $\Calls_i = \Set{(b,e) \mid b\in B_i, e \in \En_j, j = Y_i(b)}$ 
denote the set of \emph{calls} of module $A_i$ and let 
$\Returns_i = \Set{(b,x) \mid b\in B_i, x\in \Ex_j, j = Y_i(b)}$ denote the set of \emph{returns} in $A_i$.
Then, $\delta_i \subseteq (N_i \cup \Returns_i)\times(N_i \cup \Calls_i)$ is the 
\emph{transition relation} for module $A_i$.
\end{itemize}
A \emph{weighted recursive game graph} (for short WRG) is a recursive game graph, equipped with a weight function 
$w$ on the transitions.
We also refer the readers to~\cite{AlurReach} for detailed description and illustration with figures 
of recursive game graphs.
WLOG we shall assume that the boxes and nodes of all modules are disjoint.
Let $B = \bigcup_i B_i$ denote the set of all boxes, 
$N=\bigcup_i N_i$ denote the set of all nodes, 
$\En = \bigcup_i \En_i$ denote the set of all entry nodes, 
$\Ex = \bigcup_i \Ex_i$ denote the set of all exit nodes, 
$V^1 = \bigcup_i V^1_i$ (resp. $V^2 = \bigcup_i V^2_i$) denote 
the set of all places under player~1's control (resp. player~2's control), 
and $V=V^1 \cup V^2$ denote the set of all vertices.
We will also consider the special case of one-player WRGs, where 
either $V^2$ is empty (player-1 WRGs) or $V^1$ is empty (player-2 WRGs).
WLOG we will assume that the every module has a unique 
entrance (a polynomial reduction to module with many entrances to one with a single entrance was given in~\cite{AlurReach}).
The module $A_0$ is the initial module, and its entry node the starting 
node of the game.

\smallskip\noindent{\bf Configurations, paths and local history.}
A \emph{configuration} $c$ consists of a sequence $(b_1,\dots,b_r,u)$, 
where $b_1,\dots,b_r \in B$ and $u\in N$.
Intuitively,  $b_1,\dots,b_r$ denote the current stack (of modules), and $u$ is the current node.
A sequence of configurations is \emph{valid} if it does not violate the transition relation.
The \emph{configuration stack height} of $c$ is $r$.
Let us denote by $\mathbb{C}$ the set of all configurations, and let 
$\mathbb{C}_1$ (resp. $\mathbb{C}_2$) denote the set of all configurations
under player~1's control (resp. player~2's control).
A \emph{path} $\pi = \atuple{c_1, c_2, c_3, \dots}$ is a valid sequence of configurations.
Let $\rho = \atuple{c_1, c_2, \dots, c_k}$ be a valid finite sequence of configurations, such that
$c_i = (b^i_1,\dots,b^i_{d_i}, u_i)$, and the stack height of $c_i$ is $d_i$.
Let $c_i$ be the first configuration with stack height $d_i = d_k$, 
such that for every $i\leq j\leq k$, if $c_j$ has stack height $d_i$, 
then $u_j \notin \Ex$ ($u_j$ is not an exit node).
The \emph{local history} of $\rho$, denoted by $\LocalHistory (\rho)$, 
is the sequence $(u_{j_1},\dots,u_{j_m})$ such that $c_{j_1} = c_i$, 
$c_{j_m} = c_k$, $j_1 < j_2 < \dots < j_m$, and the stack height of 
$c_{j_1},\dots,c_{j_m}$ is exactly $d_i$.
Intuitively, the local history is the sequence of nodes in a module. 
Note that by definition, for every $\rho \in \mathbb{C}^*$, there exists 
$i\in \Set{1,\dots,n}$ such that all the nodes that occur in 
$\LocalHistory(\rho)$ belong to $V_i$. 
We say that $\LocalHistory(\rho) \in A_i$ if all the nodes in 
$\LocalHistory(\rho)$ belong to $V_i$.

\smallskip\noindent{\bf Global game graph and isomorphism to pushdown game graphs.}
The \emph{global game graph} corresponding to a WRG $\wrg=\atuple{A_1,\dots,A_n}$ 
is the graph of all valid configurations, with an edge $(c_1,c_2)$ between 
configurations $c_1$ and $c_2$ if there exists a transition from $c_1$ to $c_2$.
It follows from the results of~\cite{AlurReach} that every recursive game graph has 
an isomorphic pushdown game graph that is computable in polynomial time.

\smallskip\noindent{\bf Plays, strategies and modular strategies.}
A play is played in the usual sense over the global game graph 
(which is possibly an infinite graph).
A (finite) play is a (finite) valid sequence of configurations 
$\atuple{c_1, c_2, c_3, \dots}$ (i.e., a path in the global game graph).
A \emph{strategy} for player 1 is a function $\tau : \mathbb{C}^* \times \mathbb{C}_1 \to \mathbb{C}$
respecting the edge relationship of the global game graph, i.e., 
for all $w \in \mathbb{C}^*$ and $c_1 \in \mathbb{C}_1$ we have that 
$(c_1,\tau(w\cdot c_1))$ is an edge in the global game graph.
A \emph{modular strategy} $\tau$ for player~1 is a set of functions 
$\Set{\tau_i}_{i=1}^n$, one for each module, where for every $i$,
we have $\tau_i : (N_i \cup \Returns_i) ^* \to \delta_i$.
The function $\tau$ is defined as follows:
For every play prefix $\rho$ we have 
$\tau(\rho) = \tau_i(\LocalHistory(\rho))$, where $\LocalHistory(\rho) \in A_i$.
The function $\tau_i$ is the \emph{local strategy} of module $A_i$.
Intuitively, a modular strategy only depends on the local history, and 
not on the context of invocation of the module.
A modular strategy $\tau = \Set{\tau_i}_{i=1}^n$ is a \emph{finite-memory} 
modular strategy if $\tau_i$ is a finite-memory strategy for every $i\in\RangeSet{1}{n}$.
A \emph{memoryless} modular strategy is defined in similar way, 
where every component local strategy is memoryless.

\smallskip\noindent{\bf Mean-payoff  objectives and winning modular strategies.}
The \emph{modular winning strategy problem} asks if player~1 has a 
modular strategy $\tau$ such that against every strategy $\strab$ for player~2 the 
play $\pat$ given the starting node and the strategies satisfy $\LimAvg(\pat) \geq \VEC{0}$ 
(note that the counter strategy of player~2 is a general strategy).

\subsection{Undecidability for multidimensional mean-payoff objectives}
In this section we will show that the problem of deciding the existence
of modular winning strategy for player~1 in WRGs with multidimensional 
mean-payoff objectives is undecidable. 
The reduction would be from reachability games over tuples of integers.
We start by introducing these games.

\smallskip\noindent{\bf Reachability games over $\Z^k$.}
A \emph{reachability game over $\Z^k$} consists of a finite-state game graph $G$, a 
$k$ dimensional weight function $w:E \to \Z^k$, and an initial weight vector 
$\vec{\nu} \in \Z^k$.
An infinite play $\pat$ is winning for player~1 if there exists some finite 
prefix $\pat'  \sqsubseteq \pat$ such that 
$w(\pat') + \vec{\nu}=0$ and the last vertex in $\pat'$ is a player-1 vertex.

\begin{figure}[!tb]
\begin{center}
\begin{picture}(100,80)(50,-100)
\node[NLangle=0.0,Nw=32.0,Nh=16.0,Nmr=2.0](n1)(72.0,-64.0){$A_1$}
\node[Nw=32.0,Nh=16.0,Nmr=2.0](n2)(136.0,-64.0){$A_2$}
\node[Nw=4.0,Nh=4.0,Nmr=2.0](n3)(56.0,-64.0){}
\node[Nw=4.0,Nh=4.0,Nmr=2.0](n5)(88.0,-64.0){}
\node[Nw=4.0,Nh=4.0,Nmr=2.0](n6)(120.0,-64.0){}
\node[Nw=4.0,Nh=4.0,Nmr=2.0](n7)(152.0,-64.0){}

\drawedge(n5,n6){$\vec{0}$}
\drawbpedge[ELside=r](n7,34,54.6,n3,-208,55.54){$\vec{0}$}
\node[Nw=138.64,Nh=59.8,Nmr=14.95](n8)(104.91,-64.0){}
\node[Nw=4.0,Nh=4.0,Nmr=2.0](n10)(36.0,-64.0){}
\drawedge(n10,n3){$\vec{0}$}
\end{picture}
\caption{Module $A_0$}\label{fig1}
\end{center}
\end{figure}
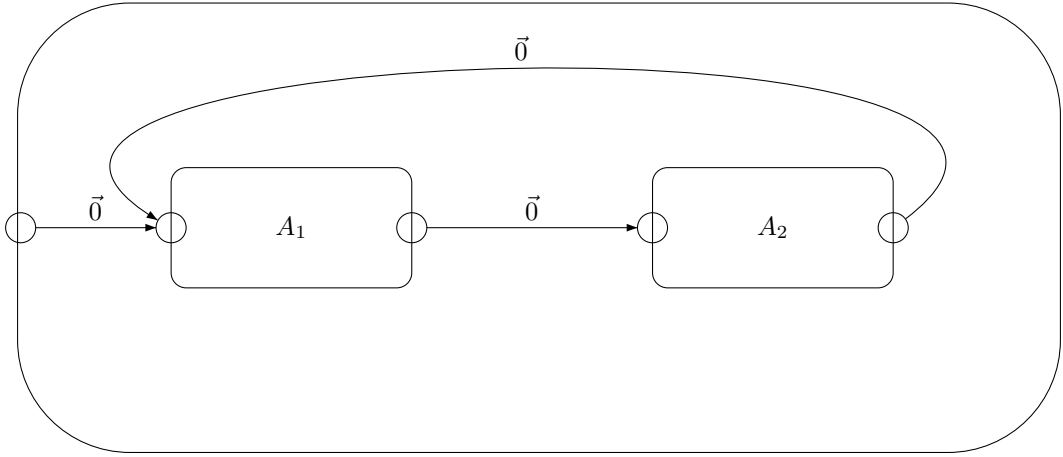

\begin{figure}[!tb]
\begin{center}
\begin{picture}(100,80)(50,-100)
\node[Nw=140.0,Nh=60.0,Nmr=15.0](A1)(105,-64.0){}
\node[Nw=4.0,Nh=4.0,Nmr=2.0](En)(35.0,-64.0){}
\node[Nw=4.0,Nh=4.0,Nmr=2.0](Ex)(175.0,-64.0){}
\node[Nw=8.0,Nh=8.0,Nmr=4.0](Loop)(105.0,-64.0){}

\drawloop(Loop){(0,0,0,0,+1,-1)}
\drawedge(En,Loop){$\vec{0}$}
\drawedge(Loop,Ex){$\vec{0}$}

\end{picture}
\caption{Module $A_1$}\label{fig2}
\end{center}
\end{figure}
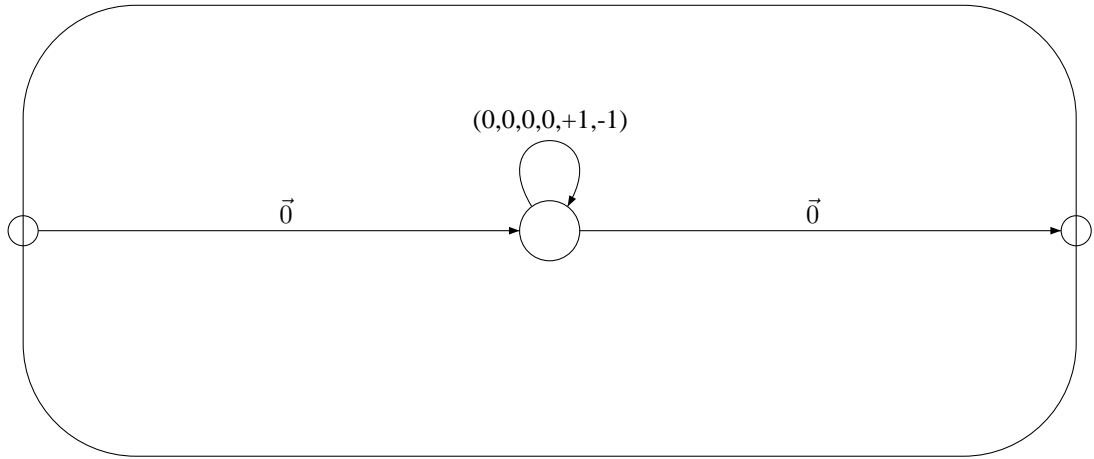

\begin{figure}[!tb]
\begin{center}
\begin{picture}(100,80)(50,-100)
\node[Nw=140.0,Nh=60.0,Nmr=15.0](A1)(105,-64.0){}
\node[Nw=4.0,Nh=4.0,Nmr=2.0](En)(35.0,-64.0){}
\node[Nw=4.0,Nh=4.0,Nmr=2.0](Ex)(175.0,-64.0){}
\node[Nw=8.0,Nh=8.0,Nmr=4.0](Vstar)(140,-64.0){$v^*$}
\node[Nw=32.0,Nh=32.0,Nmr=0](G)(88,-64.0){$G$}

\drawloop(Vstar){(0,0,0,0,-1,+1)}
\drawedge(En,G){$(+\vec{\nu},-\vec{\nu},0,0)$}
\drawedge(G,Vstar){$\vec{0}$}
\drawedge(Vstar,Ex){$\vec{0}$}

\end{picture}
\caption{Module $A_2$}\label{fig3}
\end{center}
\end{figure}
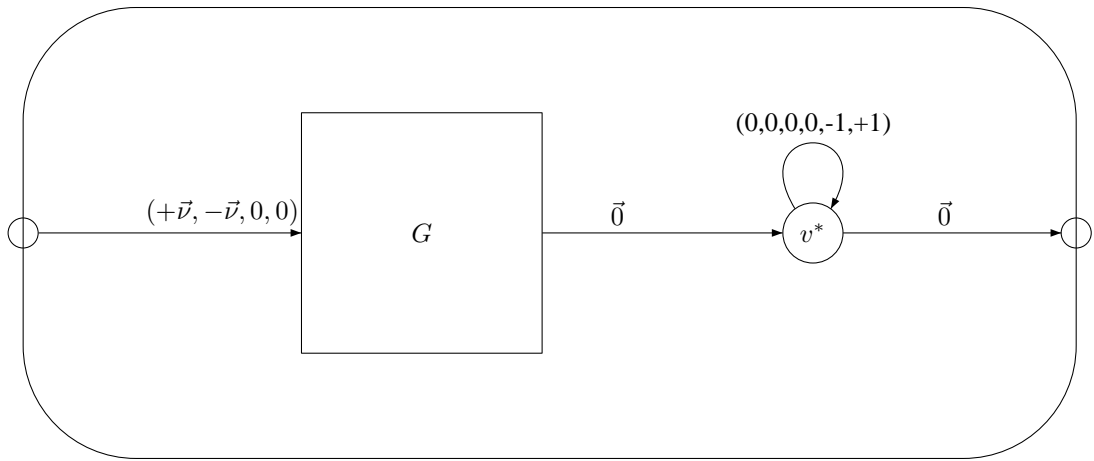

\begin{lem}\label{lemm:ReachabilityOverZIsUndecidable}
The following problem is undecidable:
Given a reachability game over $\Z^2$ and a starting vertex $v$, decide if there is a 
winning strategy $\straa$ for player~1 to ensure that for all strategies $\strab$ 
for player~2 the play $\pat(\straa,\strab,v)$ is winning for player~1.
\end{lem}
\begin{proof}
We make a simple observation that the undecidability proof for reachability games over $\Nat^2$ 
(e.g., see~\cite{AB03}) is easily extended to games over $\Z^2$. 
\end{proof}

We will present a general reduction from reachability games over $\Z^k$ to WRGs under modular strategies 
with multidimensional mean-payoff objectives of $2\cdot k + 2$ dimensions, with three modules 
(two of them with single exit, and an initial module without any exits).
Given a reachability game over $\Z^k$ with game graph $G$, weight function $w$ and initial vector 
$\vec{\nu}$, we construct a WRG graph $\wrg = \atuple{A_0,A_1,A_2}$ with a weight function of 
$2\cdot k+2$ dimensions in the following way.
\begin{itemize}
\item Module $A_0$: This module repeatedly invokes $A_1$ and $A_2$ (one call to $A_1$ and one call to $A_2$); and 
all the weights of the transitions are $0$.
\item Module $A_1$: This module has three nodes: entrance, exit and an additional one with a self-loop edge with weight $0$ in the first $2\cdot k$ dimensions, 
weight $+1$ in dimension $2\cdot k+1$ and weight $-1$ in dimension $2\cdot k+2$;
the weight of the edges from the entrance node to the additional node and from the additional node to the exit node are $0$ in
every dimension. 
All the nodes are in the control of player~1.

\item Module $A_2$: The nodes of this module are the entrance and exit nodes, 
the nodes $V$ of the reachability game $G$, and an additional node $v^*$.
The entrance node leads to the initial vertex of $G$ with edge weight 
$(\vec{\nu},-\vec{\nu},0,0)$ (i.e., the first $k$ dimensions are according to 
$\vec{\nu}$, dimensions $k+1$ to $2\cdot k$ are according to $-\vec{\nu}$, 
and the last two dimensions are $0$).
For every edge $e=(u,v)$ in $G$, there is such transition in $A_2$ with weight 
$(w(e),-w(e),-1,+1)$.
In addition, from every player-1 vertex in $V$ there is a transition to $v^*$ with weight 
$0$ in every dimension.
In $v^*$ there is a self-loop transition with weight $-1$ in dimension $2k+1$, $+1$ in dimension 
$2k+2$ and $0$ in the rest of the dimensions; and
in addition there is a transition to the exit node with weight $0$ in every dimension.
\end{itemize}
The pictorial descriptions of the modules $A_0,A_1$, and $A_2$ are shown in Figure~\ref{fig1}, Figure~\ref{fig2},
and Figure~\ref{fig3}, respectively.

\begin{Obs}\label{obs:SimpleObs}
The following observations hold:
\begin{enumerate}
\item If player-1 strategy $\tau_1$ for module $A_1$ is to never exit, then it is not a 
winning strategy (since the mean-payoff in dimension $2\cdot k+2$ will be $-1$.)
\item If for a player-1 strategy $\tau_2$ for module $A_2$, there is a play $\pat$ consistent with $\tau_2$ that does not reach $v^*$, 
then $\tau_2$ is not a winning strategy (since the mean-payoff of $\rho$ in dimension $2\cdot k+1$ will be $-1$.)
\end{enumerate}
\end{Obs}

\begin{lem}\label{lemm:IfLoosesInReachabilityThenAlsoInModuler}
If player~1 does not have a winning strategy in the reachability game over $\Z^k$, then 
there is no modular winning strategy for player~1 in $\wrg$.
\end{lem}
\begin{proof}
If player~1 does not have a winning strategy in the reachability game over $\Z^k$,
then let $\strab$ be a player-2 winning strategy for the reachability game.
We fix player-2 strategy for the modular game to be $\strab$ 
according to the local history of $A_2$ and claim that it is a winning strategy for player~2 in the
WRG against the multidimensional mean-payoff objective for player~1.
Indeed, let $\tau = \Set{\tau_1,\tau_2}$ be a player-1 modular strategy, and 
we consider the path $\pi$ which is formed by playing according to $\tau$ and $\sigma$.
By Observation~\ref{obs:SimpleObs} if $\pi$ never exit $A_1$ or never reach node $v^*$, then player~2 wins.
Otherwise, since $\strab$ is a winning strategy in the reachability game, we get that in the 
first sub-path of $\pi$ that leads from the entrance of $A_2$ to $v^*$, one of the dimensions $1 \leq i \leq 2\cdot k$ 
has a negative weight.
We note that both $\strab$ and $\tau$ are modular strategies, and thus the path $\pi$ is periodic and the mean-payoff of 
$\pi$ in dimension $i$ is negative.
To conclude, if player 2 is the winner in the reachability game, then player~1 does not have a modular
winning strategy in $\wrg$.
\end{proof}

\begin{lem}\label{lemm:IfPlayer1WinsInReachHeWinsInModular}
If player~1 has a winning strategy in the reachability game,
then there is a modular winning strategy for player~1 in $\wrg$.
\end{lem}
\begin{proof}
Let $\tau_G$ be a player-1 winning strategy for the reachability game.
By K\"{o}nig's Lemma there exists a fixed constant $n\in\Nat$ such that player~1 
can assure the reachability objective, against every player-2 strategy, 
with at most $n$ rounds.
We now derive a modular winning strategy in $\wrg$ from $\tau_G$:
\begin{itemize}
\item Module $A_1$: Follow the self-loop edge for $n$ rounds and exit.
\item Module $A_2$: Follow strategy $\tau_G$, until the weight in every dimension, according to the reachability game over $G$, 
is $0$ and a player-1 vertex was reached, and then go to $v^*$.
Let $m$ be the number of rounds played according to $\tau_G$ in the current local history of $A_2$, 
then player~1 follows the self-loop in $v^*$ for $n-m$ times and goes to the exit node.
\end{itemize}
It is easy to observe that any play according to the strategy above has a mean-payoff value of $0$ in every dimension. 
\end{proof}

From Lemma~\ref{lemm:ReachabilityOverZIsUndecidable}, 
Lemma~\ref{lemm:IfLoosesInReachabilityThenAlsoInModuler}, and 
Lemma~\ref{lemm:IfPlayer1WinsInReachHeWinsInModular}  
we obtain the following result:
\begin{thm}
\label{thm:undec}
The problem of deciding the existence of a modular winning 
strategy in WRGs with multidimensional mean-payoff objectives 
is undecidable, even for hierarchical games (i.e., games without recursive calls), 
with six dimensions, three modules and with at most single exit for each module.
\end{thm}

In view of Theorem~\ref{thm:undec} we will focus on complexity and algorithms for 
WRGs under modular strategies for one-dimensional mean-payoff objectives.

\subsection{NP-hardness}
We consider WRGs under modular strategies with one-dimensional mean-payoff objectives.
It was already shown in~\cite{CV12} that if the number of modules is not bounded,
then even if all modules have at most one exit, the problem is NP-hard
even when there is only player~1 and weights are restricted to $\Set{-1,0,1}$.
We present a similar hardness result when the number of modules are restricted
to only two, but the number of exits are not bounded.
We present a simple log-space reduction from 3SAT to WRGs with two modules.
The objective we will consider is the reachability objective, where the mean-payoff
objective is satisfied once a vertex $r$ is reached (i.e., $r$ has a self-loop 
with weight~0 and all other transitions have negative weight).

\begin{figure}
\begin{center}
\begin{picture}(50,80)(80,-100)

\node[Nw=180.0,Nh=100.0,Nmr=15.0](A1)(105,-64.0){}

\node[Nw=4.0,Nh=4.0,Nmr=2.0](En)(15.0,-64.0){}


\node[Nw=12.0,Nh=32.0,Nmr=2.0](X1A1)(43,-44.0){$A_1$}
\node[Nw=12.0,Nh=32.0,Nmr=2.0](NotX2A1)(66,-44.0){$A_1$}
\node[Nw=12.0,Nh=32.0,Nmr=2.0](X3A1)(89,-44.0){$A_1$}

\node[Nw=12.0,Nh=32.0,Nmr=2.0](X2A1)(122,-44.0){$A_1$}
\node[Nw=12.0,Nh=32.0,Nmr=2.0](NotX3A1)(145,-44.0){$A_1$}
\node[Nw=12.0,Nh=32.0,Nmr=2.0](NotX4A1)(168,-44.0){$A_1$}


\node(Bad)(156.5,-100.0){$\neg r$}
\node(Good)(156.5,-80.666){$r$}

\node[Nw=2.0,Nh=2.0,Nmr=1.0](enX1A1)(37,-44.0){}
\node[Nw=2.0,Nh=2.0,Nmr=1.0](enNotX2A1)(60,-44.0){}
\node[Nw=2.0,Nh=2.0,Nmr=1.0](enX3A1)(83,-44.0){}
\node[Nw=2.0,Nh=2.0,Nmr=1.0](enX2A1)(116,-44.0){}
\node[Nw=2.0,Nh=2.0,Nmr=1.0](enNotX3A1)(139,-44.0){}
\node[Nw=2.0,Nh=2.0,Nmr=1.0](enNotX4A1)(162,-44.0){}

\node[Nw=5.75,Nh=5.75,Nmr=2.875](exX1A1)(49,-36.0){$x_1$}
\node[Nw=5.75,Nh=5.75,Nmr=2.875](exNotX2A1)(72,-36.0){$\neg x_2$}
\node[Nw=5.75,Nh=5.75,Nmr=2.875](exX3A1)(95,-36.0){$x_3$}
\node[Nw=5.75,Nh=5.75,Nmr=2.875](exX2A1)(128,-36.0){$\neg x_2$}
\node[Nw=5.75,Nh=5.75,Nmr=2.875](exNotX3A1)(151,-36.0){$x_3$}
\node[Nw=5.75,Nh=5.75,Nmr=2.875](exNotX4A1)(174,-36.0){$x_4$}

\node[Nw=5.75,Nh=5.75,Nmr=2.875](NexX1A1)(49,-52.0){$\neg x_1$}
\node[Nw=5.75,Nh=5.75,Nmr=2.875](NexNotX2A1)(72,-52.0){$x_2$}
\node[Nw=5.75,Nh=5.75,Nmr=2.875](NexX3A1)(95,-52.0){$\neg x_3$}
\node[Nw=5.75,Nh=5.75,Nmr=2.875](NexX2A1)(128,-52.0){$x_2$}
\node[Nw=5.75,Nh=5.75,Nmr=2.875](NexNotX3A1)(151,-52.0){$\neg x_3$}
\node[Nw=5.75,Nh=5.75,Nmr=2.875](NexNotX4A1)(174,-52.0){$\neg x_4$}


\drawedge(En,enX1A1){}
\drawedge(NexX1A1,enNotX2A1){}
\drawedge(NexNotX2A1,enX3A1){}

\drawedge(exX3A1,enX2A1){}

\drawedge(exX2A1,enNotX3A1){}
\drawedge(exNotX3A1,enNotX4A1){}

\drawedge[curvedepth=-8.0](NexX3A1,Bad){}

\drawbpedge(exX1A1,0,0,enX2A1,100,60){}
\drawbpedge(exNotX2A1,0,0,enX2A1,100,45){}

\drawedge(NexX2A1,Good){}
\drawedge[curvedepth=8.0](NexNotX3A1,Good){}
\drawedge[curvedepth=10.0](NexNotX4A1,Good){}

\drawedge[curvedepth=20.0](exNotX4A1,Bad){}

\drawloop[loopdiam=4,loopangle=270.0](Good){$0$}
\drawloop[loopdiam=4,loopangle=270.0](Bad){$-1$}

\end{picture}
\caption{Module $A_1$ for 
$(x_1 \vee \neg x_2 \vee x_3)\wedge ( x_2 \vee \neg x_3 \vee \neg x_4)$}\label{fig4}
\end{center}
\end{figure}
\begin{figure}
\begin{center}
\begin{picture}(107,90)(50,-110)

\node[Nw=180.0,Nh=120.0,Nmr=15.0](A1)(105,-64.0){}

\node[Nw=4.0,Nh=4.0,Nmr=2.0](En)(15.0,-64.0){}

\node[Nw=8.0,Nh=8.0,Nmr=4.0](exX1)(195,-22.0){$x_1$}
\node[Nw=8.0,Nh=8.0,Nmr=4.0](exNotX1)(195,-34.0){$\neg x_1$}

\node[Nw=10.0,Nh=10.0,Nmr=0](X1)(135,-28.0){$x_1$}

\node[Nw=8.0,Nh=8.0,Nmr=4.0](exX2)(195,-46.0){$x_2$}
\node[Nw=8.0,Nh=8.0,Nmr=4.0](exNotX2)(195,-58.0){$\neg x_2$}

\node[Nw=10.0,Nh=10.0,Nmr=0](X2)(135,-52.0){$x_2$}

\node[Nw=8.0,Nh=8.0,Nmr=4.0](exX3)(195,-70.0){$x_3$}
\node[Nw=8.0,Nh=8.0,Nmr=4.0](exNotX3)(195,-82.0){$\neg x_3$}

\node[Nw=10.0,Nh=10.0,Nmr=0](X3)(135,-76.0){$x_3$}

\node[Nw=8.0,Nh=8.0,Nmr=4.0](exX4)(195,-94.0){$x_4$}
\node[Nw=8.0,Nh=8.0,Nmr=4.0](exNotX4)(195,-106.0){$\neg x_4$}

\node[Nw=10.0,Nh=10.0,Nmr=0](X4)(135,-100.0){$x_4$}

\drawedge(En,X1){}
\drawedge(En,X2){}
\drawedge(En,X3){}
\drawedge(En,X4){}

\drawedge(X1,exX1){}
\drawedge(X1,exNotX1){}

\drawedge(X2,exX2){}
\drawedge(X2,exNotX2){}

\drawedge(X3,exX3){}
\drawedge(X3,exNotX3){}

\drawedge(X4,exX4){}
\drawedge(X4,exNotX4){}
\end{picture}
\caption{Module $A_0$}\label{fig5}
\end{center}
\end{figure}
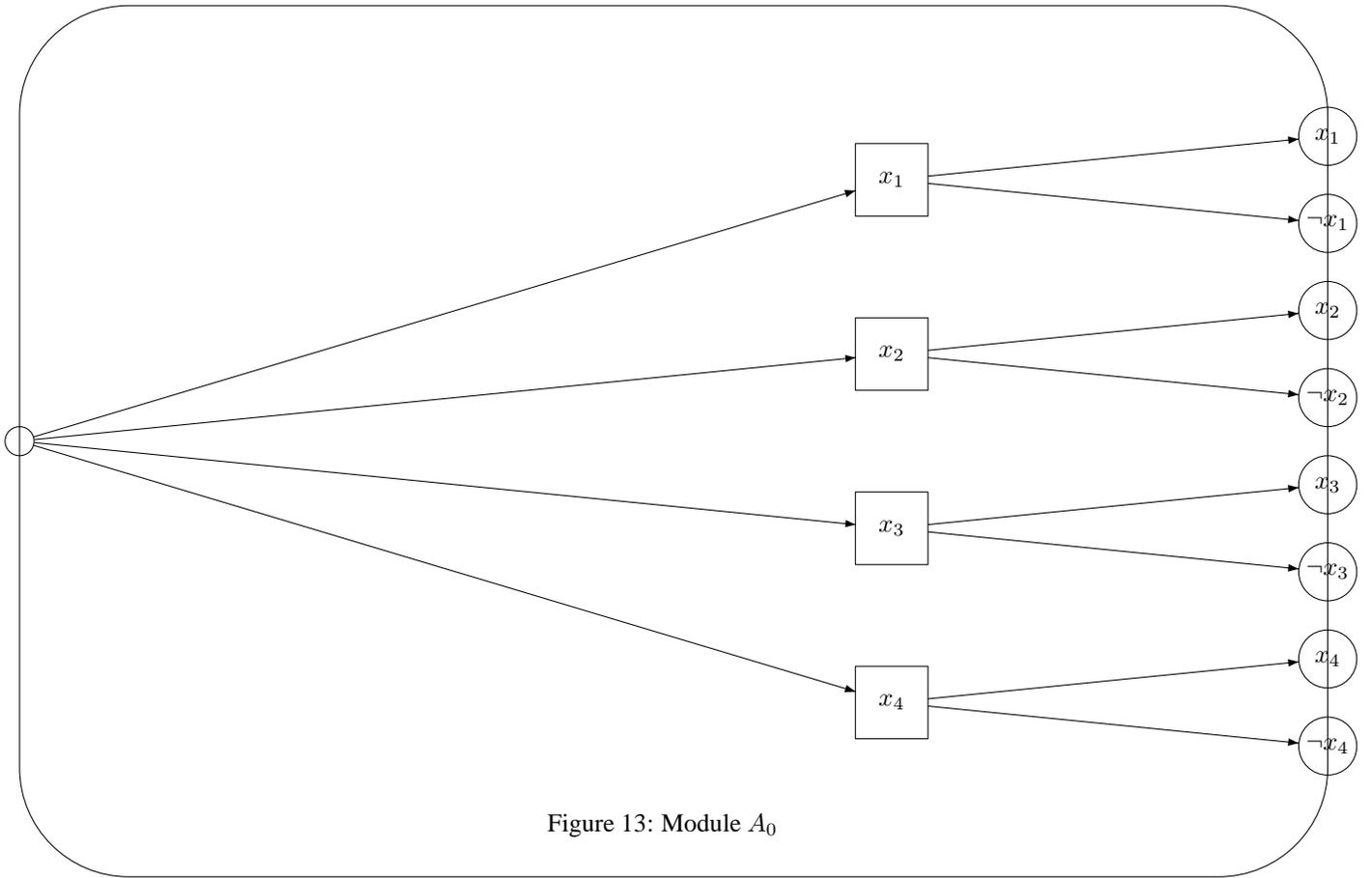

\Heading{The reduction.}
For a 3SAT formula $\varphi(x_1,\dots,x_n) = \bigwedge_{i=1}^m C_i$ we construct a WRG with two modules, namely $A_0$ and $A_1$.
\begin{itemize}
\item Module $A_1$: The module has $2n$ exits namely, $\Ex_{x_1},\Ex_{\neg x_1},\dots,\Ex_{x_n},\Ex_{\neg x_n}$, an entrance node that is owned by player 2, and $n$ player-1 nodes $x_1,\dots,x_n$.
From the entrance node there is a transition $(\En,x_i)$, for $i=1,\dots,n$;
and from every node $x_i$ there is one transition to $\Ex_{x_i}$ and one transition to $\Ex_{\neg x_i}$.
Intuitively, a modular strategy for player 1 is to decide on a True/False value for every $x_i$.
\item Module $A_0$: This is the initial module; it consists of $m$ gadgets $C_1,\dots,C_m$ (note that these are gadgets and not modules), and two sink states, namely $r$ and $\neg r$, where $r$ is the reachability objective.
A gadget $C_i = y^1_i \vee y^2_i \vee y^3_i$ consists of three sub-gadgets, namely, $y^1_i, y^2_i, y^3_i$;
gadget $y^j_i$ invokes module $A_1$ and the exits $(\Set{\Ex_{x_1},\Ex_{\neg x_1},\dots,\Ex_{x_n},\Ex_{\neg x_n}} \setminus \Set{y^j_i, \neg y^j_i})$ 
of $A_1$ leads to the good sink node $r$, the exit $y^j_i$ leads to gadget $C_{i+1}$ (or to node $r$ if $i=m$), and the exit $\neg y^j_i$ leads to sub-gadget $y^{j+1}_i$ 
(or to the bad sink node $\neg r$ if $j = 3$).
\end{itemize}
The reduction is illustrated in Figure~\ref{fig4} and Figure~\ref{fig5}.
It is an easy observation that player 1 has a modular winning strategy iff the formula $\varphi$ is satisfiable.

\begin{thm}\label{thm:np-hard}
The decision problem of existence of modular winning strategies in WRG's
with one-dimensional mean-payoff objectives is NP-hard even for 
WRG's with two modules and weights restricted to $\Set{0,-1}$.
\end{thm}


\subsection{Algorithm for one-dimensional dimensional mean-payoff objectives} 
Given the undecidability result, we focus on WRGs with one-dimensional 
mean-payoff objectives, and given the hardness results for either unbounded 
number of modules or unbounded number of exits, our goal is to present an algorithm
that runs in polynomial time if both the number of modules and the number of exits are bounded.
For the rest of this section we denote the number of game modules by $\NumModules$, the number of exits and boxes 
(in the entire graph) by $\NumExits$ and $\NumBoxes$, respectively, and 
by $n$ and $m$ the maximal size of $|V_i|$ and $|\delta_i|$ 
(number of vertices and transitions) respectively that a module has. 
If $\NumModules$, $\NumExits$ and $W$ (the maximal absolute weight) are bounded, 
then our algorithm runs in polynomial time. 
We first present a theorem from~\cite{CV12} that will be useful in our result
and then present the notion of cycle-free memoryless modular strategy.

\begin{thm}[\cite{CV12}]\label{thm:memoryless-modular}
Given a WRG $\wrg$ with a one-dimensional weight function, if there is a modular 
winning strategy for the objective $\LimAvg$, then there is a memoryless
modular winning strategy.
\end{thm}

\smallskip\noindent{\bf Negative-cycle-free memoryless modular strategy.}
A player-1 memoryless modular strategy $\tau$ is called 
\emph{negative-cycle-free memoryless modular strategy} if in the recursive graph $\wrg^\tau$ 
there are no proper cycles $C$ with negative weights, i.e., $w(C)<0$.

\smallskip\noindent{\bf Signature of a negative-cycle-free memoryless modular strategy.}
The \emph{signature} of a negative-cycle-free memoryless modular strategy 
$\tau = \Set{\tau_i}_{i=1}^\NumModules$ is an $\NumModules$-tuple of function 
$\Sig(\tau) = \Set{\Sig_i : \Ex_i \to \Z\cup \Set{-\omega,+\infty}}_{i=1}^\NumModules$ such 
that for an exit node $x$ in module $A_i$ we have $\Sig_i(x) = z$ if
\begin{itemize}
\item $z\in\Z$ and the non-decreasing path with the minimal weight in $\wrg^\tau$ from $\En_i$ to $x$ 
(in the same stack height) has weight $z$.
\item $z = +\infty$ and there is no non-decreasing path in $\wrg^\tau$ from  $\En_i$ to $x$.
\item $z = -\omega$ and for every integer $j$ there is a non-decreasing path in $\wrg^\tau$ from $\En_i$ to $x$ (at the same stack height), 
with weight at most $j$.
\end{itemize}

The next lemma demonstrates an important property of signature functions.
\begin{lem}\label{lemm:SigIsBounded}
Let $\ell = (\NumModules \cdot n) ^{\NumModules \cdot \NumExits + 1}$;
let $\tau$ be a negative-cycle-free memoryless modular strategy; and 
let $W$ denotes the maximal weight (in absolute value) that occur in $\wrg$.
Then $\Sig(\tau)$ has the following property:
\begin{quote}
For every $i\in\RangeSet{1}{\NumModules}$, the image (range) of $\Sig_i$ is $\{-\omega,+\infty\}\cup (\Z \cap [-W\cdot \ell,W\cdot \ell])$
\end{quote}
\end{lem}
\begin{proof}
We fix the strategy $\tau$ in $\wrg$, and obtain the player-2 recursive game 
graph $\wrg^\tau$.
To show the result we need to prove that if there is a path 
(in $\wrg^\tau$) from $\En_i$ (the entrance of $A_i$) to 
$x \in  \Ex_i$ with weight less than $-W \cdot \ell$, 
then for every $r\in\Z$ there exists a path from 
$\En_i$ to $x$, consistent with $\tau$, and with weight less than $r$;
and that it is impossible that the path with the minimal weight from 
$\En_i$ to $x$ has weight at least $W\cdot \ell + 1$.

The proof is as follows: 
let $\pi$ be the shortest path in $\wrg^\tau$ 
from $\En_i$ to $x$ with weight $w(\pi) < -W\cdot \ell$ 
(note that $\pi$ corresponds to a play consistent with $\tau$).
Since $w(\pi) < -W\cdot \ell$, it must be that $|\pi| > \ell$, 
therefore $\pi$ must have a pumpable pair of paths, and 
since $\pi$ is the shortest path from $\En_i$ to $x$ with such weight, 
the weight of the pumpable pair must be strictly negative, and 
thus we can construct paths from $\En_i$ to $x$ with arbitrary small weights.

Similarly, we show that if there is a path from $\En_i$ to $x$, 
then there is a path with weight at most $W\cdot \ell - 1$.
Towards contradiction, let $\pi$ be the path with minimal weight between 
$\En_i$ and $x$ and $w(\pi)\geq W\cdot \ell$ and $\pi$ is the shortest path 
with minimal weight.
As $|\pi|\geq \ell$ it follows that it has a pumpable pair of paths $P$.
If $w(P) > 0$ or $w(P) < 0$, 
then we get a contradiction to the fact that $\pi$ has minimal weight (either by omitting 
$P$ if $w(P)<0$ or pumping $P$ arbitrarily if $w(P)>0$).
If $W(P) = 0$, then we get a contradiction to the assumption that $\pi$ is the shortest path by
simply omitting $P$. 
The desired result follows.
\end{proof}

\Heading{Feasibility of signature.}
We say that a function $\Sig : \Ex\to \Set{-\omega,+\infty}\cup \Z$ is \emph{feasible} if there is a negative-cycle-free memoryless modular strategy $\tau$ such that $\Sig(\tau) = \Sig$.

\begin{lem}\label{lemm:SignatureFeasability}
Given a threshold vector $\VEC{\nu}\in (\Set{-\omega,+\infty}\cup \Z)^\NumExits$,
we can verify in 
$(\NumModules \cdot n)^{O(\NumModules \cdot \NumExits^2)} \cdot W^{O(\NumExits)}$
time if there exists 
a feasible signature function $\Sig : \Ex\to \Set{-\omega,+\infty}\cup\Z$ 
such that $\Sig \geq \VEC{\nu}$ (i.e., for every $x\in\Ex$ we have $\Sig(x) \geq \nu_x$).
\end{lem}
\begin{proof}
By Lemma~\ref{lemm:SigIsBounded} we may assume that the input is restricted for 
$\VEC{\nu}\in (\Set{-\omega,+\infty}\cup (\Z\cap[-W\ell,+W\ell]))^\NumExits$
and $\Sig : \Ex \to \Set{-\omega,+\infty}\cup(\Z\cap[-W\ell,+W\ell])$.
The proof of the lemma will use the idea of signature verification 
games.

\smallskip\noindent{\em The signature verification games.}
For a recursive game $\wrg$ and a function $\Sig : \Ex \to \Set{-\omega,+\infty}\cup \Z$ 
we construct $\NumModules$ game modules $G_1,\dots,G_\NumModules$, 
such that $G_i$ is formed from the module $A_i$ by replacing every box $b$, 
that invokes module $A_j$ and its $k$-th return node leads to node $v_k$, with a player-2 
node $v_b$ and edges $(v_b,v_k)$ with weight $\Sig_j(\Ex_k)$. 
Intuitively every game module is like a finite-state game with thresholds for exit
vertices. 
We first prove a claim related to signature verification games.

\smallskip\noindent{\bf Claim.}
For every game module $G_i$ there exists a strategy $\tau_i$ that satisfies $\Sig$, i.e., it assures:
\begin{itemize}
\item every path in $G_i^{\tau_i}$ from $\En_i$ to $\Ex_j$ has weight at least $\Sig_i(\Ex_j)$; and
\item there are no cycles with negative weight in $G_i^{\tau_i}$;
\end{itemize}
if and only if there exists a feasible signature function $\Sig '$ such that 
$\Sig ' \geq \Sig$. 

\smallskip\noindent{\em Proof of claim.} 
We prove both the directions of the claim. We start with the left to the right direction.
By Theorem~\ref{thm:memoryless-modular} such strategies $\Set{\tau_i}_{i=1}^\NumModules$ exist 
iff there exist memoryless strategies $\Set{\tau^*_i}_{i=1}^\NumModules$ that satisfies the above.
Clearly, $\tau^*$ is also a modular strategy.
In addition, for every path $\pi$ in $\wrg^{\tau^*}$, the path does not 
contain negative proper cycles (and hence, $\tau^*$ is a negative-cycle-free strategy), 
and the path does not violates the constraints according to $\Sig$.
The proof is by a straightforward induction on the additional stack height of $\pi$.
Hence we have $\Sig(\tau^*) \geq \Sig$. 
The other direction is simpler.
Clearly if there exists a negative-cycle-free modular strategy $\tau$ 
such that $\Sig(\tau) \geq \Sig$, then $\tau_i$ satisfies both items for every game module $G_i$.
This proves the desired claim.

The results of~\cite[Lemma 31]{CV-Corr} provides an algorithm that decides 
if for a given function $f:\Ex \to \Set{-\omega,+\infty}\cup(\Z\cap[-W\ell,+W\ell])$ and 
a game module $G_i$ there is a memoryless strategy that satisfies $f$;
this is done by solving a (finite-state) mean-payoff game with one-dimensional 
objective with weights at most $2\cdot n\cdot W\cdot \ell$.
Hence, we can sequentially go over all the functions 
$f:\Ex \to \Set{-\omega,+\infty}\cup(\Z\cap[-W\ell,+W\ell])$ such that 
$f \geq \VEC{\nu}$ and check if $f$ is satisfiable.
By the claim a signature $\Sig \geq \VEC{\nu}$ exists if and only if such $f$ was found.

\emph{Complexity.} The complexity analysis is as follows:
there are $(2\cdot W\cdot \ell + 2)^\NumExits$ functions to verify;
in the verification process we solve $\NumModules$ mean-payoff games with weights at most 
$2\cdot n\cdot W\cdot \ell$ and at most $n$ vertices and $m$ edges; and
every mean-payoff game can be solved in $O(m\cdot n^2\cdot W\cdot \ell)$ 
time~\cite{BCDGR11}.
Thus the the overall complexity is 
\[
O(n^2 \cdot m \cdot (W\cdot \ell)^{\NumExits+1} \cdot \NumModules) = 
O(n^{\NumModules\cdot \NumExits^2 + \NumModules\cdot \NumExits + \NumExits + 3} \cdot m \cdot
\NumModules^{\NumModules\cdot \NumExits^2 + \NumModules\cdot \NumExits + \NumExits + 2} \cdot 
W^{\NumExits +1})= 
(\NumModules \cdot n)^{O(\NumModules \cdot \NumExits^2)} \cdot W^{O(\NumExits)}
\]
The desired result follows.
\end{proof}

\smallskip\noindent{\bf Reduction from modular games to signature problem.}
Intuitively, for a given WRG $\wrg$, we would like to construct a new WRG $\wrg '$, 
such that player 1 is the winner in $\wrg$ iff there exists a feasible signature in 
$\wrg'$ with certain properties.
We construct $\wrg '$ in the following way:
Let $(A_1,\dots,A_\NumModules)$ be the modules of $\wrg$, then we construct the modules 
$(A '_1,\dots,A '_\NumModules)$ from $(A_1,\dots,A_\NumModules)$ as follows:
\begin{itemize}
\item Add $\NumModules$ exit nodes $x_1,\dots,x_\NumModules$ for every module. 
\item For every box node $b$, in module $A_j$, if $b$ invokes module $A_i$, then 
for all $k\neq i$, the exit $x_i$ is connected (by an edge with weight $0$) 
to the exit $x_i$ in the module $A_j$, and if $k=i$, then the exit leads 
to a sink state (and the weight of the self-loop is positive).
\item W.l.o.g we assume that all the entrances are player-2 nodes, and we add edges 
with zero weight from each entrance to all the new exits $x_1,\dots,x_\NumModules$.
\end{itemize}

We note that the number of exits $\NumExits'$ in $\wrg'$ is $\NumExits + \NumModules^2$.
The following lemma establishes winning in $\wrg$ and properties of signature function 
in $\wrg'$.

\begin{lem}\label{lemm:WinIffThereIsASignature}
Player 1 has a memoryless modular winning strategy in $\wrg$ iff there is a feasible signature 
$\Sig$ in $\wrg '$ such that for every module $A '_i$ we have $\Sig_i(x_i) \geq 0$. 
\end{lem}
\begin{proof}
We first prove the direction from left to right.
Let $\tau$ be a memoryless modular winning strategy (and therefore also negative-cycle-free) in $\wrg$. 
We note that $\tau$ is a modular negative-cycle free strategy also for $\wrg '$.
We claim that (the feasible signature function) $\Sig = \Sig(\tau)$ satisfies $\Sig_i(x_i) \geq 0$.
Indeed, if $\Sig_i(x_i) < 0$, then by the construction of $\wrg'$, there is a play $\rho$ from 
$\En_i$ to $\En_i$ (at an higher stack height) with negative weight, that is consistent with $\tau$.
Since $\tau$ is a modular strategy we get that $\rho^\omega$ is a play with a negative mean-payoff 
that is consistent with $\tau$, which contradicts the assumption that $\tau$ is a winning strategy.

To prove the converse direction, let $\tau$ be a memoryless negative-cycle-free strategy in $\wrg '$ such that $\Sig(\tau) = \Sig$.
We note that $\tau$ is a modular strategy also for $\wrg$ and we claim that it is a winning strategy for $\wrg$.
Indeed, let $\wrg^\tau$ be the player-2 game according to $\tau$;
if in $\wrg^\tau$ there is a path with negative mean-payoff then either
\begin{itemize}
\item there is a proper cycle in $\wrg^\tau$ with negative weight, which contradicts the assumption that $\tau$ is negative-cycle-free strategy; or
\item there is a non-decreasing cycle $\wrg^\tau$ with negative weight.
If this is the case then for some module $A_i$ there is a non-decreasing path in $\wrg^\tau$ from $\En_i$ to $\En_i$ with negative weight, and thus in $\wrg '^\tau$ there is a path with negative weight from $\En_i$ to $x_i$ and therefore $\Sig_i(x_i) < 0$, in contradiction to the assumption. 
\end{itemize}
The desired result follows.
\end{proof}

\begin{thm}\label{thm:SolveBoundedRecursiveGamesWith}
Given a WRG $\wrg$ with a one-dimensional mean-payoff objective,
whether player~1 has a modular winning strategy can be decided in $(n\cdot \NumModules)^{O(\NumModules^5 + \NumModules\cdot \NumExits^2)}\cdot W^{O(\NumModules^2 +\NumExits)}$ 
time.
\end{thm}
\begin{proof}
We first construct the modular game graph $\wrg '$ and then we check if there is a signature function $\Sig$ such that $\Sig_i(x_i) \geq 0$ for every $i\in\RangeSet{1}{\NumModules}$.
The correctness and complexity follows from Lemma~\ref{lemm:WinIffThereIsASignature} and Lemma~\ref{lemm:SignatureFeasability}.
\end{proof}

\subsection{Hardness for fixed parameter tractability}
Given Theorem~\ref{thm:SolveBoundedRecursiveGamesWith} (algorithm to solve
in polynomial time when $\NumModules$ and $\NumExits$ are fixed) an interesting question 
is whether it is possible to show that WRGs under modular strategies 
is fixed parameter tractable (i.e., to obtain an algorithm that runs in time 
$O(f(\NumModules,\NumExits) \cdot \text{poly}(n,m,W))$).
We show the hardness of fixed parameter tractability, again by a reduction 
from parity games, implying that fixed parameter tractability would imply
the solution of the long-standing open problem of fixed parameter tractability
of parity games.

\smallskip\noindent{\bf Parity games to mean-payoff games with large weights.}
In~\cite{Jur98} a reduction of finite-state parity games to finite-state 
mean-payoff games was presented, and the weights for the mean-payoff game 
used were $\set{(-n)^0,(-n)^1,\dots,(-n)^i,\dots,(-n)^k}$, where $k$ is the number of 
priorities of the parity function.
The reduction was a $O(k\cdot n \cdot \log n)$ time reduction.

\smallskip\noindent{\bf The reduction.}
Given a finite state mean-payoff game $G$ with $n$ vertices and weights 
$(-n)^0,(-n)^1,\dots,(-n)^i,\dots,(-n)^k$ we construct a recursive game 
graph $\wrg = \Tuple{A_0,P_1,\dots,P_k,N_1,\dots,N_k}$ with $2\cdot k +1$ modules in the following 
way.

\begin{itemize}
\item 
The $P_i$ modules: all the nodes  in the $P_i$ modules have out-degree $1$ (so the owner is irrelevant), 
and all the modules have only one exit.
In module $P_1$ the out-edge of the entrance node leads to the exit node and has weight $+n$ 
(equivalently, it has a path with length $n$ to the exit node, and the weight of each edge in the path is $+1$).
For $i > 1$, the module $P_i$ invokes $n$ times the module $P_{i-1}$ and goes to the exit node.

\item
The $N_i$ modules: all the nodes  in the $N_i$ modules have out-degree $1$ (so the owner is irrelevant), 
and all the modules have only one exit.
In module $N_1$ the out-edge of the entrance node leads to the exit node and has weight $-n$ 
(equivalently, it has a path with length $n$ to the exit node, and the weight of each edge in the path is $-1$).
For $i > 1$, the module $N_i$ invokes $n$ times the module $N_{i-1}$ and goes to the exit node.

\item The $A_0$ module: $A_0$ is formed from the vertices of the finite state game graph $G$, 
and every transition $(u,v)$ in $G$, with weight $(-n)^i$ is replaced by a transition from $u$ to a 
box $b$ and by a transition from the return node of $b$ to $v$ (both with weight $0$), 
where $b$ invokes $P_i$ if $i$ is even, and invokes $N_i$ if $i$ is odd.

\end{itemize}

\begin{rem}\label{rem:WeightsOfModules}
The path from the entrance of module $P_i$ (resp. $N_i$) to its exit has weight $n^i$ (resp. $-(n^i)$).
\end{rem}
\begin{proof}
The proof is by a trivial induction on $i$.
\end{proof}

We observe that all strategies in $\wrg$ are modular strategies, and that a modular winning strategy in 
$\wrg$ is a winning strategy in $G$, 
and vice versa. 
We have the following result.

\begin{thm}\label{thm:parity-modular-reduction}
Given a finite-state parity game $G$ with $n$ vertices and priority function 
of $k$-priorities, we can construct in polynomial time a WRG $\wrg$ with $2\cdot k+1$ 
modules, with $O(k\cdot n)$ nodes and weights restricted to $\Set{-1,0,+1}$ 
such that a vertex $v$ is winning for player~1 in the parity game iff there is a 
modular winning strategy in $\wrg$ with $v$ as the initial node.
\end{thm}
